\documentclass[12pt,a4paper,final]{article}

\usepackage[color,notref,notcite]{showkeys}
\definecolor{refkey}{gray}{.8}   %1=white
\definecolor{labelkey}{gray}{.4} %0=black

\usepackage[margin=27mm]{geometry}
\usepackage{amsfonts,amsmath,bef_alex,bm,accents}
\usepackage{datetime,comment}
\usepackage[utf8]{inputenc}

\def\FG{\bm} \def\BS{\bm}
\def\emph#1{{\bfseries\itshape #1}}
\renewcommand{\div}{\mathop{\mathrm{div}}}
\newcommand{\curl}{\mathop{\mathrm{curl}}}
\newcommand{\trace}{\mathop{\mathrm{Tr}}\nolimits}
\def\rme{\ee} % Euler's number \exp(1)

\newcommand{\diag}{\mathop{\mathrm{diag}}\nolimits}
\newcommand{\ma}{\mathrm{ma}}
\newcommand{\qs}{\mathrm{qs}}
\newcommand{\RHObeta}{\wh\rho_\beta}
\renewcommand{\hat}{\wh}
\renewcommand{\tilde}{\wt}

\def\HH{{\bm{\mathfrak h}}}

\def\SP#1#2{\langle #1 | #2\rangle}
\newcommand{\dyad}[2]{ | #1 \rangle \;\! \langle #2 |}
\def\SPP#1#2%{\big\langle} #1 \;\!\big| #2 \big\rangle}
 {\bm\langle\hspace{-0.3em}\bm\langle \;\!#1 \;\! \bm|\!\bm|\;\! 
                     #2\;\! \bm\rangle\hspace{-0.3em}\bm\rangle}

\def\KB{k_\rmB}
\newcommand{\eq}{\mathrm{eq}}
\def\epsilon{\eps}

\newcommand{\ADJ}{\calT}% _{\hat{\sigma}}}

\newcommand{\block}[4]{\bma{@{\,}c@{\;\;}c@{\,}} #1&#2\\[-0.15em] #3&#4\ema}

\newcommand{\bigSet}[2]{\big\{\, #1 \,\big|\, #2 \,\big\} }
\def\oti{{\otimes}}

\newcommand{\RRR}{\mathfrak R}
\newcommand{\CCC}{\mathcal C}
\newcommand{\DDD}{\mathcal D}
\newcommand{\MM}[2]{\mathcal{M}_{ #2 }^{ #1 }}
\newcommand{\MMM}{\mathcal M}
\newcommand{\KKK}{\mathcal K}
\newcommand{\KK}[2]{\mathcal{K}_{ #2 }^{ #1 }}
\newcommand{\SSS}{\mathcal  S}
\newcommand{\LLL}{\mathcal L}
\newcommand{\Herm}[1]{\mathrm{Herm} ( #1 )}
\newcommand{\YYY}{\mathcal Y}

\numberwithin{equation}{section}
\begin{document} 

\title{An entropic gradient structure for\\ Lindblad equations and \\
       couplings of quantum systems\\  to macroscopic models%
\thanks{Original title:  An entropic gradient structure for Lindblad
  equations and GENERIC for quantum systems
coupled to macroscopic models}
\thanks{The research of M.M. was supported by ERC via AdG 267802
  AnaMultiScale, and A.M. was partially supported 
  by DFG via SFB 787 Nanophotonics (subproject B4).}}
\author{Markus Mittnenzweig$^1$ and Alexander Mielke$^{1,2}$\\
{\normalsize $^1$ Weierstraß-Institut für Angewandte Analysis und Stochastik,
Berlin }\\[-0.3em]
{\normalsize $^{2}$ Institut f\"ur Mathematik, Humboldt-Universit\"at zu
Berlin\qquad\qquad\ \ \mbox{}} }

\date{October 13, 2016; revised February 12, 2017}
\maketitle

\section{Introduction}
\label{s:Intro}

In many situations the evolution of quantum systems is dictated not
only by the system Hamiltonian but also by dissipative effects, i.e.\
the time-dependent density matrix $\rho(t)$ satisfies a
dissipative evolution equation of the form
\begin{equation}
  \label{eq:I.LE}
  \dot \rho = \frac{\ii}{\hbar}\,\big[\rho, H \big]\ + \LLL \rho,
\end{equation}
where we will set $\hbar=1$ in the sequel. The dissipative part $\LLL$
of this Lindblad equation has to be the generator of semigroup
consisting of completely positive operators, which is enforced by
the structure of quantum mechanics. In this work we additionally ask
for the \emph{condition of detailed balance} with respect to the
equilibrium state $\RHObeta$ in the sense of Alicki \cite{Alic76DBCN}
and \cite{KFGV77QDBK}, by which we mean that $ \calL \RHObeta=0$ and
that $\calL^*$ is symmetric with respect to the weighted operator
scalar product $(A,B) \mapsto \trace(A^* B\RHObeta)$, see
\eqref{eq:DBC}. We call such $\calL$ shortly DBC Lindbladians with
respect to $\RHObeta$.  Equations like \eqref{eq:I.LE} where already
derived in the seminal work \cite{Davi74MME} as the weak
coupling limit of a given quantum system with a large heat bath. It
was observed in \cite{Spoh78EPQD} that this class of models satisfies
the detailed-balance condition (DBC) and that the relative entropy
(also called free energy)
\[
\calF(\rho)= \trace\big( \rho\,(\log\rho -\log\RHObeta) \big)
= \trace\big( \rho\,(\log\rho + \beta H ) \big) + \log Z_\beta
\]
is a Liapunov function, i.e.\ it decays along solutions. Here
$\beta>0$ is a suitable inverse temperature and 
\[
\RHObeta = \frac1{Z_\beta}\:\ee^{-\beta H} \quad \text{with }Z_\beta=
\trace\big( \ee^{-\beta H}\big)
\]
is the thermal equilibrium. 

The aim of this work is to show that \eqref{eq:I.LE} can be written as
a damped Hamiltonian system, namely 
\begin{equation}
  \label{eq:I.DHS}
  \dot\rho = \big( \bbJ(\rho)- \bbK(\rho) \big) \rmD\calF(\rho),
\end{equation}
where the operator $\bbJ(\rho):\xi \mapsto
\frac{\ii}{\beta}[\rho,\xi]$ generates a Poisson bracket, while the
operator $\bbK(\rho)$ should be purely dissipative, i.e.\
$\bbK(\rho)=\bbK(\rho)^*\geq 0$.  We will call such dissipative
operators simply \emph{Onsager operators}, because of Onsager's
fundamental work in \cite{Onsa31RRIP}. We continue to use $\rmD$ for
the differential of functionals, i.e.\ $\langle
\rmD\calF(\rho),v\rangle :=\lim_{h\to 0} \frac1h\big(\calF(\rho{+}h v)
- \calF(\rho)\big)$.

Thus, our aim is the construction of an Onsager operator $\bbK$ which
generalizes the Wasserstein operator $\bbK_\text{Wass}(u):\mu \mapsto
-\div\big(\rho \nabla \mu\big)$ for the Fokker-Planck equation and the  
Markov operator $\bbK_\text{Mv}(p)$ for jump processes constructed in
\cite{Miel11GSRD,Maas11GFEF,ErbMaa12RCFM,Miel13GCRE}, cf.\ Section
\ref{su:ClassMP}.  The crucial
point is that $\bbK(\rho)$ has to depend on $\rho$ in a very specific
way to obtain the relation
\[
- \bbK(\rho) \big(\log \rho + \beta H\big) = \LLL \rho,
\]
where the right-hand side is linear in $\rho$. In the Fokker-Plank
equation this is achieved by the chain rule $u \nabla\big( \log u +
V)= \nabla u + u \nabla V$, and for jump processes it follows from 
$\Lambda(a,b)(\log a - \log b)=a-b$,  see Section \ref{su:ClassMP}. 

For quantum systems first steps in this direction were done in 
\cite{Otti10NTQM,Otti11GTDQ,Miel13DQMU,CarMaa14A2WM,Miel15TCQM}. They
involve the use of the Kubo-Mori operator 
\[
\CCC_\rho:L(\HH)\to L(\HH);\quad  A\mapsto \CCC_\rho A:= \int_0^1 \rho^s
\,A\, \rho^{1-s} \dd s, 
\]
which satisfies for all $Q\in L(\HH)$ the fundamental relation
\begin{equation}
  \label{eq:I.Mira}
  \CCC_\rho\,\big[Q,\log \rho\big]  \;= \;\big[ Q,\rho\big],
\end{equation}
which we will call the miracle relation. It goes back to
\cite[Eqn.\,(2.17), p.\,139]{Kubo59ASMT} and was put into a more
general context in \cite{Wilc67EOPD}, see (2.5) there. Obviously, we
need a generalization allowing for commutators of the form
$[Q,\log\rho + \beta H]$, which is nontrivial if $\beta H\neq \alpha
\bm1_\HH$. In \cite{CarMaa14A2WM} the infinite-temperature case $\beta
H=0$ is treated, while in
\cite{Otti10NTQM,Otti11GTDQ,Miel13DQMU,Miel15TCQM} the nonlinear terms
$\rho\mapsto \CCC_\rho[Q,H]$ are admitted.

Here, we show that in the general case with DBC it is
possible to find a suitable $\bbK$ in a rather natural way. The
starting point is a tensor-product representation of Lindblad
operators. We set $\HH_1=\HH$ and choose an arbitrary second Hilbert
space $\HH_2$ and assume that $\HH_1$ and $\HH_2$ are
finite-dimensional. For an arbitrary Hermitian $\bbQ \in
L(\HH_1\oti\HH_2)$ and a $\wh\sigma\in L(\HH_2)$ with
$\wh\sigma=\wh\sigma^*>0$ one sees that 
\begin{equation}
  \label{eq:I.LSigma}
  \LLL \rho= - \trace_{\HH_2}\Big( \big[ \bbQ, [\bbQ, \rho\oti
\wh\sigma]\big] \Big)
\end{equation}
is indeed a Lindblad operator.  Moreover, it can be shown easily that
this $\LLL$ satisfies the DBC with respect to $\RHObeta$ if the
commutation relation
\[
\big[\, \bbQ \,,\, \RHObeta \otimes \wh\sigma\,\big] \ = \ 0
\]   
holds, see Section \ref{su:Lind.CompForm}. Under this condition it is
now straightforward to show the following generalization of the
miracle identity:
\[
\CCC_{\rho\oti\wh\sigma} \big[\, \bbQ\, , \, (\log\rho {+} \beta H) 
\otimes \bm1_{\HH_2}\,\big] = \big[ \, \bbQ\, , \,\rho\otimes \wh\sigma\,\big].
\]
Indeed, it suffices to use the fact that $\bbQ$ also commutes with $
\log(\RHObeta\oti \wh\sigma)=-\beta H\oti \bm1_{\HH_2} +
\bm1_{\HH_1}\oti \log\wh\sigma$ and then apply the classical miracle
identity \eqref{eq:I.Mira}, see Theorem \ref{th:MiracleIdentity} for
the details. Now we can define the Onsager operator 
\begin{equation}
  \label{eq:I.bbK}
  \bbK(\rho) \xi = \trace_{\HH_2}\Big( \big[ \bbQ,
\CCC_{\rho\oti\wh\sigma} [\bbQ, \xi\oti \bm1_{\HH_2}] \Big),
\end{equation}
which is a symmetric and positive semidefinite operator and  satisfies
the desired relation
\begin{equation}
\label{eq:I.MiraExtended}
-\bbK(\rho)\big( \log \rho + \beta H\big) =
- \trace_{\HH_2}\Big( \big[ \bbQ, [\bbQ, \rho\oti
\wh\sigma]\big] \Big)= \LLL \rho,
\end{equation}
which provides the entropic gradient structure for the
dissipative terms. Thus,there is a full analogy to classical Markov
processes, for which it was shown in \cite{MiPeRe14RGFL} that the DBC
implies the existence of an entropic gradient structure.

In Theorem \ref{th:CompRepDBC} we show that every DBC Lindbladian with
respect to $\RHObeta$ can be written in the form \eqref{eq:I.LSigma}
with $\HH_2=\HH=\HH_1$ and either $\wh\sigma = \RHObeta$ or
$\wh\sigma=\RHObeta^{-1}$. Here we rely on the classical
characterization of DBC Lindblad operators in
\cite{Alic76DBCN,KFGV77QDBK}.

In Sections \ref{s:GENERIC} and  \ref{se:ExaAppl} we consider a few
applications and discuss 
the general problem of modeling the interaction of a macroscopic
system described by a state variable $z\in Z$. We show that it is
possible to set up a coupled system in the framework of GENERIC, which
is an acronym for ``General Equations for Non-Equilibrium Reversible
and Irreversible Coupling''. This framework is based on an energy
functional $\calE$, an entropy functional $\calS$, a Poisson operator
$\bbJ$, and an Onsager operator $\bbK$ such that the evolution is
\[
\binom{\dot\rho}{\dot z} = \bbJ_\text{coupl}(\rho,z) 
\binom{\rmD_\rho \calE(\rho,z)}{\rmD_z \calE(\rho,z)}
+ \bbK_\text{coupl}(\rho,z)\binom{\rmD_\rho \calS(\rho,z)}{\rmD_z \calS(\rho,z)}.
\]
This is complemented by the fundamental \emph{non-interaction
  conditions} $\bbJ\rmD\calS \equiv 0  \equiv \bbK\rmD\calE$ that
defines a thermodynamically consistent system with energy conservation
and entropy production. The typical choice for the functionals is
\[
\calE(\rho,z)=\trace(\rho H) + E(z) \quad \text{and} \quad 
\calS(\rho,z)=-\KB \trace\big(\rho \log \rho\big)+ S(z).
\]

To describe the coupling of the quantum system
with the variable $z$ it is essential to model the different
dissipation mechanisms separately, which we do by the minimal
building blocks $\SSS_{W}$ and $\MM{\beta}{Q}$ for Lindblad
operators, where $Q$ must satisfy $[Q,H]=\omega Q$ for some $\omega
\in \bbR$.  The associated Onsager operators $\KK{\beta}{Q}$ can then be
obtained from the construction \eqref{eq:I.bbK} by choosing
\[
\bbQ= Q^* \otimes \block0100+ Q\otimes \block0010 \quad \text{and}
\quad  \wh\sigma = \block{\ee^{\beta \omega/2}}{0}{0}{\ee^{-\beta \omega/2}}. 
\]
With the above choice for $\calE$ we have $\rmD_\rho \calE(\rho,z)=H$,
which forces us to use fixed eigenpairs $(\omega,Q)$. However, we may
assume that the effective coupling temperature may depend on $z$ and
may differ for different coupling mechanisms. Hence, a typical Onsager
operator for the coupling between a quantum system for $\rho$ and a
classical system for $z\in Z$ may have the form
\[
\bbK_\text{coupl}(\rho,z) = \sum_{m=1}^M
\bma{@{}cc@{}}\KK{\wt\beta_m(z)}{Q_m}(\rho) & 
{\langle \Box,b_m(z)\rangle_Z \KK{\wt\beta_m(z)}{Q_m}(\rho)H}\\[0.4em]
{\SPP{\KK{\wt\beta_m(z)}{Q_m}(\rho)\Box}{H}b_m(z)}&
{\SPP{\KK{\wt\beta_m(z)}{Q_m}(\rho) H}{H}b_m(z)\oti b_m(z)} \ema,
\]
where $\langle \cdot,\cdot\rangle_Z$ denotes the dual pairing between
$Z^*$ and $Z$, while $\SPP{A}{B}$ denotes the (Hilbert-Schmidt) scalar
product  $\trace(A^*B)$ for operators, see Section \ref{su:Setup}. 
For an application to the thermodynamically consistent modeling of the
Maxwell-Bloch system as considered in \cite{JoMeRa00TNGO,Duma05GEMB}
we refer to Section \ref{su:MaxwellB}, where the macroscopical
variable $z=(\bfE,\bfH)$ contains the electric and the magnetic
field. For more applications, 
also involving the coupling to drift-diffusion equations we refer to
\cite[Sec.\,5]{Miel15TCQM}.  

For the sake of notational simplicity and to avoid technical
complications, we have restricted ourselves to the finite-dimensional
case $ \mathrm{dim} \, \HH < \infty$. Many results have an immediate
generalization to the infinite-dimensional case,
e.g.\ the relations  \eqref{eq:I.Mira} to
\eqref{eq:I.MiraExtended}. However, the full 
characterization of all DBC Lindbladians, which
is given in Theorem \ref{th:CompRepDBC} or Proposition
\ref{pr:BuildBlock}, would require much more delicate considerations
and thus remains an open question. 

Note added in proof: After this work was finished the authors became
aware of the parallel and totally independent  work
\cite{CarMaa16?GFEI}, which has some overlap concerning the
construction of an entropic gradient structure for Lindblad equations
with detailed balance.

\section{Dissipative quantum mechanics}
\label{s:CoupledModel}

\subsection{General notations and setup}
\label{su:Setup}

Here we recall the standard theory and introduce our notation.  The
quantum mechanical system is described by states in a complex Hilbert
space $\HH$ with scalar product $\SP a b$. For a Hamiltonian operator
 $H \in \Herm{\HH}$ (the set of 
Hermitian operators on $\HH$) the associated Hamiltonian dynamics is
given via the Schr\"odinger equation $   \dot \psi = -\ii H \psi$, 
which has the solution $\psi(t)=\ee^{-\ii t H} \psi(0)$. 

To couple a quantum system to a macroscopic one we need to describe it
in statistical terms  using the density matrices
\[
\rho \in \RRR_N := \set{ \rho \in L ( \HH )  }{ \rho=\rho^* \geq 0,\
  \trace \, \rho =1 }. 
\]
Each $\rho \in \RRR_N$ has the representation 
\begin{equation}
  \label{eq:RepresRho}
  \rho = \sum_{j=1}^N r_j \,\dyad{\psi_j}{\psi_j}, 
\end{equation}
where $r_j\geq 0$, $\sum_1^N r_j=1$, and $\set{ \psi_j }{j=1,...,N}$
is an orthonormal set. Note that in our notation $(\dyad\psi\phi)\,
a := \SP{\phi}{a} \psi$ and $(\dyad\psi\phi) A=\dyad\psi{A^*\phi}$.  

On operators we define the Hilbert-Schmidt scalar product $\SPP{A}{B}
= \trace (A^* B) $ satisfying the following identities, which will be
used below without further notice:
\begin{align*}
&\SPP AB =  \SPP{B^*}{A^*} = \ol{\SPP BA}, \qquad\ \quad \SPP{\lambda A}{\mu
  B}= \ol\lambda \mu \SPP AB,\\
&\SPP A{BC}= \SPP{AC^*}B = \SPP{B^*A}C, 
 \quad %\\ %
\SPP A {[B,C]} = \SPP {B^*}{[C,A^*]} =\SPP{C^*}{[A^*,B]} ,
\end{align*}
where $\lambda,\mu\in \C$ and $A,B,C\in L (\HH )$.

\subsection{The Lindblad equations} 
\label{su:QuantumMech}

Using the Schr\"odinger equation the evolution of $\rho$ is given via
the  Liouville-von Neumann equation
\begin{equation}
  \label{eq:RhoEvol}
\dot \rho = -\ii [ H, \rho ], \quad \text{where } [\rho, H ]:=  \rho
H{-} H\rho . 
\end{equation}

For open systems, dissipative versions of the Hamiltonian
Liouville--von Neumann equation are used. The most general linear
master equation preserving complete positivity is the well-known Lindblad
equation
\begin{equation}
  \label{eq:LindEvol}
\dot \rho = -\ii [ H, \rho ] +\LLL\rho \quad \text{with }\LLL A=
\sum_{n,m=1}^{N^2-1} a_{n,m} \big( [Q_n, A Q_m^*]+[Q_n A , Q_m^*]\big) ,
\end{equation}
where $Q_n$ are arbitrary operators in
$L(\HH):=\mathrm{Lin}(\HH,\HH)$, and $(a_{n,m})$ is a Hermitian
positive semi-definite matrix. Note that $\LLL$ in the Lindblad
equation is evaluated only on $\rho\in \RRR$, while we prefer to
define $\LLL$ as an operator mapping from all of $L(\HH)$ into
$L(\HH)$. It is easily seen that every $\LLL$ is a $*$-operator, i.e.\
it satisfies $\LLL(A^*)=(\cal L A)^*$.  The set of all Lindblad
operators forms a cone of real dimension $ ( N^2 {-} 1 ) ^2$ in the
set of linear operators from $L(\HH)$ into itself. Characterizing the
steady states and the dynamics of a general Lindblad operator remains
a field of ongoing research \cite{BaNaTh08AQSG1,BauNar08AQSG2,BauNar12SSSQ}.
\bigskip

In this work we are mainly interested in Lindblad operators satisfying
the detailed balance condition (DBC), shortly called DBC
Lindbladians. We follow the definition in
\cite{Alic76DBCN,KFGV77QDBK} and refer to \cite{AJPP06TNEQ} for a
discussion of other versions of the DBC. Let $\LLL^\ast$
denote the adjoint of $\LLL$ defined via $\SPP{\LLL ^\ast A}{B} = \SPP
{A}{\LLL B}$. The \emph{condition of detailed balance} with respect to
the equilibrium state $\RHObeta$ is defined via the relation
\begin{equation}
\label{eq:DBC}
\text{(DBC)} \qquad \left\{\begin{aligned} & \LLL \RHObeta=0, \\
     &\SPP{\LLL^\ast (A) }{B \RHObeta} =
   \SPP{A}{\LLL^\ast (B) \RHObeta} 
\quad \text{for all }A,B\in L(\HH), 
\end{aligned}\right.
\end{equation}
i.e.\ $\calL^*$ is symmetric with respect to the weighted scalar
product $(A,B)\mapsto \SPP A{B \RHObeta}$. The characterization of all
operators lying in the class of DBC Lindbladians for a fixed
$\RHObeta$ is given in \cite[Eqn.\,(20)]{Alic76DBCN}. See also
\cite{JaPiWe14EFQD} for a modern characterization of DBC Lindblad
operators, while there the main goal is a large deviations theory for
the asymptotics for $t \to \infty$.  We will derive a new and compact
representation of these operators in Section
\ref{su:Lind.CompForm}. Clearly, the DBC is equivalent to
\begin{equation}
  \label{eq:DBC-variant}
  \LLL ( A\RHObeta ) = \LLL^\ast
(A) \RHObeta \quad \text{for all } A\in L(\HH).
\end{equation}

\begin{comment}
Moreover by introducing the weighted scalar product
\[
\SPP AB_{\RHObeta} := \SPP{A}{B\RHObeta} 
\]
we can state the DBC in an even simpler way. Let
$\LLL^{\ast,\RHObeta} $ be the $\RHObeta$-adjoint of $\LLL$ defined
via $\SPP {\LLL^\ast A}{B}_{\RHObeta} = \SPP
A{\LLL^{\ast,\RHObeta}B}_{\RHObeta}$. Then \eqref{eq:DBC} simply
means that $\LLL^\ast= \LLL^{\ast,\RHObeta}$.

ANMERKUNG: Ich finde die Notation $\LLL^{\ast,\RHObeta}$ sehr
unschön, aber mir ist gerade keine gute eingefallen. Man könnte statt
$\LLL$ vielleicht auch $L$ verwenden, aber dann hätte man weiter
unten bei $\SSS_W$ und $\MMM_Q$ immer noch das Problem. Außerdem habe
ich den Operator $\calM_{X,Y}A=XAY$ herausgenommen, weil, wenn ich
mich nicht täusche, er weiter unten nicht mehr verwendet
wird. Übrigens war für mich bei den Rechnungen mit dem Adjungierten
folgende Eigenschaft hilfreich, die mir anfangs nicht bewusst war:
Lindblad-Operatoren sogenannte $\ast$-Abbildungen, d.h. sie erfüllen
$\LLL  (  A^\ast  )  =  (  \LLL  (  A  ) 
 ) ^\ast$
\end{comment}

Before discussing the general form of \emph{all} DBC Lindbladians, we
will construct minimal building blocks. They  
are useful in their own right for the modeling of dissipative
couplings as discussed in Sections \ref{s:GENERIC} and
\ref{se:ExaAppl}. The main observation is that the property of
detailed balance with respect to $\RHObeta=\frac1{Z_\beta}\ee^{-\beta
  H}$ involves operators $Q$ having the property $\RHObeta Q
(\RHObeta)^{-1} = \mu Q$, which can be characterized by the following
elementary result.

\begin{lemma}\label{le:Comm.Rel} For $\omega \in \R$, $H\in \Herm{\HH}$ and 
  $Q\in L ( \HH )$ we have the equivalences: 
 \begin{align}
  \label{eq:Comm.Rel} 
\begin{aligned}
\text{(i) } \ [Q,H]=\omega Q \quad \Longleftrightarrow& \quad 
\text{(ii) } \ \exists \,  \beta\neq 0:\ 
   \ee^{- \beta H} Q \ee^{ \beta H} = \ee^{\beta \omega} Q
\\
 \Longleftrightarrow& \quad \text{(iii) } \  \forall\, \gamma\in \R: \   
\ee^{- \gamma H} Q \ee^{ \gamma H} = \ee^{\gamma \omega} Q.
\end{aligned}
 \end{align}
\end{lemma}

We note that operators with the commutator property (i) can easily be
constructed when using the spectral representation of the Hamiltonian
$H$, namely 
\[
H= \sum_{n=1}^N \eps_n \dyad{h_n}{h_n}, \quad \text{ and hence } 
\RHObeta= \frac1{Z_\beta} \sum_{n=1}^N
\ee^{-\beta \eps_n} \dyad{h_n}{h_n} \text{ with }Z_\beta=\sum_{n=1}^N
\ee^{-\beta \eps_n}. 
\]
For a given eigenvalue $\eps_n$ we define the spectral projector $P_n$
via 
\[
P_n = \sum_{k:\eps_k=\eps_n}\dyad{h_k}{h_k},\quad \text{giving }
P_n=P_n^2 = P_n^* \text{ and } P_nH=HP_n = \eps_n P_n.
\]
Now we can take any operator $V\in L ( \HH )$ and choose spectral
projectors $P_n$ and $P_m$. Then 
\[
Q=P_n V P_m \quad \text{satisfies} \quad [Q,H]= (\eps_m {-} \eps_n) Q.
\]
We emphasize that this relation is linear in $Q$, so that a general
$Q$ satisfying $[Q,H]=\omega Q$ may have the form 
\[
Q= \sum_{(n,m):\;\eps_m{-}\eps_n=\omega } P_n V_{n,m} P_m,
\]
thus possibly more than two energy levels $\eps_k$ may be
involved. This is trivial for the case $\omega=0$ but may also occur
in the case $\omega\neq 0$, see Example \ref{suu:NontrivialCase}. 
 
We introduce the spectrum $\Omega(H)$ of the map $A\mapsto [A,H]$ and the set
$\mathfrak E(H)$ of eigenpairs via  
\begin{align*}
& \Omega(H):=\mathrm{spec}([\,\cdot\,,H]) = \bigSet{\eps_m{-}\eps_n}{ \eps_n,\eps_m
  \in \text{spec}(H)},  
\\
&\mathfrak E(H):=\bigSet{(\omega, Q)\in \R\ti L ( \HH )}{ [Q,H]=
  \omega Q}. 
\end{align*} 
Let us further define the multiplicities of the eigenspaces via 
\[
  d_\omega = \text{dim} \bigSet{ Q \in L(\HH ) }{ [Q,H]=\omega Q }. 
\] 
If $H$ has only one-dimensional eigenspaces and no pairs of
eigenvalues $\eps_m,\eps_n$ ($m \neq n$) have equal differences
$\eps_m {-} \eps_n =\eps_{m'} {-} \eps_{m'}$, then $d_\omega = 1$ for
$\omega \neq 0$ and $d_0 = N$. However, in the most degenerate case
$H=0$ we find $d_0=N^2$ and $d_\omega = 0$ for all $\omega$. The
following result provides the building blocks for all DBC Lindbladians
with respect to $\RHObeta=\frac1{Z_\beta}\ee^{-\beta H}$.

\begin{proposition}[Building blocks $\SSS_W$ and 
$\MM{\beta}{Q}$]\label{pr:BuildBlock}
Let $H$ and $\RHObeta$ be given as above. 

(a) Consider any $W\in \Herm{\HH}$ with $[W,H]=0$, then
the operator $\SSS_W$ defined by 
\[
\SSS_W A \; := \; [W,AW]+ [WA,W]= \big[W,[A,W]\big]
\]
is a Lindblad operator satisfying $\SSS_W=\SSS_W^\ast$ and
the DBC for $\RHObeta$.

(b) Consider any pair $(\omega,Q)\in \mathfrak E(H)$, then the operator
$\MM{\beta}{Q}$ defined via 
\[
\MM{\beta}{Q} A \; := \; \ee^{\beta \omega/2} \big( [Q,AQ^*]+[QA,Q^*]\big)+
\ee^{-\beta \omega/2} \big([Q^*,AQ]+[Q^*A,Q]\big)
\]
is a DBC Lindbladian for $\RHObeta$.

(c) Every Lindbladian $\LLL$ satisfying the DBC
\eqref{eq:DBC} can be written in the form
\[
\LLL= \sum_{j=1}^J \SSS_{W_j} + \sum_{m=1}^M \MM{\beta}{Q_m} , \ \text{
  where }\left\{ \begin{aligned} 
& W_j=W_j^*, \ (0,W_j)\in \mathfrak E(H), \text{ and }\\
& (\omega_m,Q_m)\in  \mathfrak E(H) \text{ with }\omega_m>0. 
\end{aligned}\right. 
\] 
The numbers $J$ and $M$ of necessary terms is bounded by $ J\leq
d_0{-}1$ and $M\leq \sum_{\omega\in \Omega(H)\setminus\{0\}}
d_\omega$.
\end{proposition}
\begin{proof} Part (a): Obviously, $\SSS_W$ is a special case of
  $\LLL$ in \eqref{eq:LindEvol} by choosing $a_{1,1}=1$ and $Q_1=W$
  and $a_{n,m}=0$ for $(n,m)\neq (1,1)$, so it is a Lindblad operator.
  We also see that the DBC $\SSS^{\RHObeta}_W = \SSS_W$ holds, since
  $\RHObeta W \RHObeta^{-1} =W$ by using Lemma \ref{le:Comm.Rel}.

  Part (b): It is obvious that $\MM{\beta}{Q}$ has the form of $\LLL$ in
  \eqref{eq:LindEvol} with $Q_1=Q$, $Q_2=Q^*$,
  $a_{1,1}=\ee^{\beta\omega/2}$, and $a_{2,2}=\ee^{-\beta\omega/2}$,
  while all other terms are $0$.  Moreover, $\left(\MM{\beta}{Q}\right)^{\RHObeta}$ can
  be calculated explicitly by using Lemma \ref{le:Comm.Rel} and
  $\RHObeta Q^*\RHObeta^{-1}= \ee^{-\beta\omega} Q^*$, so the DBC
  follows.

  Part (c): From \cite[Eqn.\,(2.16)-(2.20)]{KFGV77QDBK} (where $L_s$
  corresponds to our $\LLL^*$) we know that every DBC Lindbladian
  with respect to $\RHObeta$ can be written as
\begin{equation}
\label{eq:calL.gen}
\LLL A = \sum_{k,j=1}^N D_{kj} \big( [X_{kj} A, X_{kj}^*] +
[X_{kj},A X_{kj}^*]\big),
\end{equation}
where $D_{kj}\geq 0$ and $X_{kj}\in \C^{N\ti N}$ satisfy the
conditions (with $\wh r_j=\ee^{-\beta \eps_j}/Z_\beta>0$) 
\begin{equation}
\label{eq:KF*conds}
\begin{aligned}
& \text{(i) }D_{kj}\wh r_j = D_{jk}\wh r_k  \text{ and }
X_{kj}^*=X_{jk} \ \text{ for } \wh r_j\neq \wh r_k;\quad \text{(ii) }
\SPP{X_{kj}}{X_{lm}}=\delta_{kl}\delta_{jm};\\  
& \text{(iii) } X_{kj}^*=X_{kj} \text{ for }\wh r_j= \wh r_k ;\quad 
 \text{(iv) } \RHObeta \, X_{kj} (\RHObeta)^{-1} = \frac{\wh r_k}{\wh
   r_j} \, X_{kj}.  
\end{aligned}
\end{equation}
We decompose the set $I=\{1,..,N\}^2$ into $I_{\neq}:=\set{(k,j)\in
  \{1,..,N\}^2}{\wt r_j\neq \wh r_k}$ and $I_=:=\set{(k,j)\in
  \{1,..,N\}^2}{\wt r_j = \wh r_k}$. For $(j,k)\in I_{\neq}$
the second condition in (i) gives $X_{jk}=X_{kj}^*$, while (iv) and Lemma
\ref{le:Comm.Rel} imply $(\eps_k{-}\eps_j,X_{kj}) \in \mathfrak
E(H)$. Using now the first condition in (i) as well and setting 
$Q_{kj}= \ee^{\beta(\eps_k{-}\eps_j)/4} X_{kj}$, we find the
relation 
\[
D_{kj} \big( [X_{kj} A, X_{kj}^*] +
[X_{kj},A X_{kj}^*]\big)+ D_{jk} \big( [X_{jk} A, X_{jk}^*] +
[X_{jk},A X_{jk}^*]\big)= \MM{\beta}{Q_{kj}}A. 
\]

For $(k,j)\in I_=$ conditions (iii) and (iv) yield $X_{kj}= X_{kj}^*$
and $[X_{kj},H]=0$. Thus,
\[
D_{kj} \big( [X_{kj} A, X_{kj}^*] +
[X_{kj},A X_{kj}^*]\big) = \SSS_{W_{kj}} A \text{ with
}W_{kj}=\sqrt{D_{kj}} X_{kj}. 
\]
In summary, we find that $\LLL$ in \eqref{eq:calL.gen} can be written
in the form 
\[
\LLL= \sum_{(k,j)\in I_=} \SSS_{W_{kj}} + \sum_{(k,j)\in I_{\neq},\,
  j<k} \MM{\beta}{Q_{kj}}, 
\]
which is the desired result. 
\end{proof}

Note that the representation of $\LLL$ in terms of the Kraus
operators $Q_{n}$ in \eqref{eq:LindEvol} is not
unique. Correspondingly, our representation in terms of $\SSS_{W_j}$
and $\MMM_{Q_j}$ is not unique. Moreover, for a minimal
representation one may ask for additional orthogonality conditions.
We also remark that the operators $\SSS_W $ can be obtained from
$\MMM_{Q} $ as a special case allowing $\omega=0$ and asking for
$Q=Q^*$. More precisely, if $(0,Q)\in \mathfrak E(H)$ then also
$(0,Q^*)$ and $(0, \frac12(Q{+}Q^*))$ lie in $\mathfrak
E(H)$. In particular, we have 
\[
(0,Q)\in \mathfrak E(H) \text{ and } Q=Q^* \quad \Longrightarrow \quad
\MM{\beta}{Q}= 2 \SSS_{\beta,Q}= \SSS_{\beta,\sqrt 2 \:Q}. 
\] 
Moreover, $(\omega,Q)\in \mathfrak E(H)$ if and only if $({-}\omega,
Q^*)\in \mathfrak E(H)$ and $\MM{\beta}{Q}=\MM{\beta}{Q^*}$. Thus,
Proposition \ref{pr:BuildBlock}(c) tells us that all DBC Lindbladians
can be written in the form
\begin{equation}
  \label{eq:LindDBC}
\LLL \rho = \sum_{n=1}^N \MM{\beta}{Q_n} \rho
   \quad \text{where }  (\omega_n,Q_n) \in \mathfrak E(H) \text{ and
   }\omega_j\in   \Omega(H). 
\end{equation}
We will see specific examples in Sections \ref{su:Lind.Exa} and
\ref{se:ExaAppl}. The above representation in terms of the building blocks is
especially useful for modeling, while the next section provides a
form that is more elegant and compact.

\subsection{A compact form of all DBC Lindblad operators}
\label{su:Lind.CompForm}

In this section we will write Lindblad operators and our building
blocks \ref{pr:BuildBlock} in another way. The basic idea is to write
them as the partial trace of a double commutator on a larger
space. This will prove useful in Section \ref{se:EntrGSQME} when
writing down entropic gradient structures for the Lindblad equations
with DBC.  In what follows $\HH_{1}$ and $\HH_{2}$ denote
finite-dimensional Hilbert spaces, and for $A\in L(\HH_1)$ and $B\in
L(\HH_2)$ we denote the tensor product by $A \oti B \in L ( \HH_{1}
) \oti L ( \HH_{2} ) =L(\HH_1 \oti \HH_2)$. We also introduce the
notion of a partial trace
\[
\trace_{\HH_{2}}: L(\HH_1)\oti L(\HH_2) \to L(\HH_1) \ \text{
  defined via } \trace_{\HH_{2}} (A\oti B) = \trace_2(B) A 
\]
on direct products and extended by linearity to the whole space $L (\HH_1 \oti \HH_2 )$. Here $\trace_2$ is the trace in $L(\HH_2)$. 

The next result shows that all Lindblad operators on $L(\HH_1)$
can be written in a compact form by a double commutator on
$L(\HH_1)\oti L(\HH_2)$ and a partial trace. Moreover, this form
allows for a simple criterion for the DBC with respect to the
equilibrium $\RHObeta$.

\begin{proposition}[Compact representation for $\LLL$]
\label{pr:calL.CR} Consider two finite dimensio\-nal Hilbert spaces
  $\HH_1$ and $\HH_2$. Assume that $\bbQ\in \Herm { \HH_1 \oti \HH_2 } $, $\wh\sigma \in \Herm{\HH_2}$ and
  $\wh\sigma\geq0$. Then
  \begin{equation}
    \label{eq:Lind.CompForm}
    \LLL(\rho)=-\trace_{\HH_{2}}\Big(\big[\mathbb{Q}, 
  [\mathbb{Q},\rho\oti\wh\sigma] \big]\Big)
  \end{equation}
is a Lindblad operator in $L(\HH_1)$, i.e.\ the generator of a
completely positive semigroup.

If in addition $\bbQ $ and $\wh\sigma$ satisfy the commutator
relation 
\begin{equation}
  \label{eq:CommQSigma}
\big[ \mathbb{Q},\RHObeta\oti\wh\sigma \big] =0 ,  
\end{equation}
then $\LLL$ satisfies the DBC with respect to $\RHObeta$.
\end{proposition}
\begin{proof}
Since $\wh\sigma\geq0$ we can write
$\wh\sigma=\sum_{j=1}^{J}\sigma_{j} e_{j}\oti\overline{e}_{j}$
with $\sigma_{j}\geq0$. Let us define 
\[
   Q_{kl}=\left\langle e_{k}\left|\mathbb{Q}\right|e_{l}\right\rangle
_{\HH_{2}}=\trace_{\HH_{2}} ( \mathbf{1}_{\HH_{1}}\oti (
e_{l}\oti\overline{e}_{k} ) \mathbb{Q} )  \text{ giving } \bbQ =
\sum_{k,l} Q_{kl}\otimes(e_l\oti \ol e_k).
\]
Then $Q_{kl}=Q_{lk}^*$ and using $\trace_{\HH_2}\big(B\otimes
\big((e_k\oti \ol e_l)\oti(e_m\oti\ol e_n)\big)\big)=
\delta_{kn}\delta_{lm} B $ we find 
\begin{align}
 \label{eq:calL.bbQ.sigma}
\LLL\rho  & =\sum_{k,l=1}^{J}\big(2\sigma_{l}Q_{kl}\rho
Q_{lk}-\sigma_{k}\{ Q_{kl}Q_{lk},\rho\}\big) 
%\\ & 
=\sum_{k,l=1}^{J}\sigma_{l} \big( [Q_{kl}\rho,Q_{kl}^* ]+ [Q_{kl},\rho
    Q_{kl}^*] \big) . 
\end{align}
which is clearly of Lindblad form. 

The commutation relation \eqref{eq:CommQSigma} immediately implies
$\LLL \RHObeta=0$, which is the first relation in the DBC
\eqref{eq:DBC}. The second relation  is written in terms of the dual operator
$\LLL^*$ that takes the form 
\[
\LLL^* ( A ) =\trace_{\HH_{2}} ( \mathbf{1}_{\HH_{1}}\oti\wh\sigma\,\left[\mathbb{Q},\left[\mathbb{Q},A\oti
      \mathbf{1}_{\HH_{2}}\right]\right] ) . 
\]
We have to show $\trace \big( (\LLL^* A )^* B\RHObeta ) \big) =
\trace ( A^*\LLL^* ( B ) \RHObeta ) $. Using $(\LLL^*
A)^*=\LLL^*(A^*)$ the left hand side is equivalent to
\begin{align*}
\trace_{\HH_1} ( \LLL^* ( A^* ) B\RHObeta ) 
&
=\trace_{\HH_{1}\oti\HH_{2}} ( \mathbf{1}_{\HH_{1}}\oti\wh\sigma
  \,\left[\mathbb{Q},\left[\mathbb{Q},A^*\oti\mathbf{1}_{\HH_{2}}\right]
  \right]\, ( B\RHObeta ) \oti\mathbf{1}_{\HH_{2}} ) \\ 
 &
 =\trace_{\HH_{1}\oti\HH_{2}} ( \,\left[\mathbb{Q},\left[\mathbb{Q},
       A^*\oti\mathbf{1}_{\HH_{2}}\right]\right]\,
    ( B\oti\mathbf{1}_{\HH_{2}} ) 
   \RHObeta\oti\wh\sigma ) .
\end{align*}
Again using the commutator condition \eqref{eq:CommQSigma} we obtain,
for all $\bbA$, the identity
\begin{align*}
\left[\mathbb{Q},\mathbb{A} ( \RHObeta\oti 
   \wh\sigma ) \right]
&
=+\left[\mathbb{Q},\mathbb{A}\right]\RHObeta\oti\wh\sigma
+\mathbb{A}\left[\mathbb{Q},\RHObeta\oti\wh\sigma\right]
% \\& 
=\left[\mathbb{Q},\mathbb{A}\right]\RHObeta\oti\wh\sigma,
\end{align*}
which we use twice, namely once with
$\mathbb{A}=B\oti\mathbf{1}_{\HH_{2}}$ and once with 
 $\mathbb{A}=\left[\mathbb{Q},B\oti\mathbf{1}_{\HH_{2}}\right]$.
Thus, we can move the $\mathbb{Q}$ operators to the right and obtain
\begin{align*}
\mathrm{Tr}_{\HH_{1}} ( \LLL^* ( A^* ) B\RHObeta ) 
&
=\mathrm{Tr}_{\HH_{1}\oti\HH_{2}} ( A^*\oti
  \mathbf{1}_{\HH_{2}}\,\left[\mathbb{Q},\left[\mathbb{Q}, ( B\oti
        \mathbf{1}_{\HH_{2}} ) \right]\right] 
  \RHObeta\oti\wh\sigma ) \\ 
 &
 =\mathrm{Tr}_{\HH_{1}} ( A^*\mathrm{Tr}_{\HH_{2}}
   ( \left[\mathbb{Q},\left[\mathbb{Q}, ( B\oti\mathbf{1}_{\HH_{2}} ) 
       \right]\right]\mathbf{1}\oti\wh\sigma ) \RHObeta ) 
\\ &
 = \mathrm{Tr}_{\HH_{1}} ( A^*\LLL^* 
   ( B ) \RHObeta ) , 
\end{align*}
which is the desired DBC. 
\end{proof}

The above result shows that the commutator relation
\eqref{eq:CommQSigma} is crucial for the study of DBC Lindbladians. In
the following lemma we give an alternative characterization which will
be useful later, when studying the associated gradient structures.

\begin{lemma}[Equivalent commutation relation]\label{le:CommQSigma}
Consider $\bbQ\in \Herm{\HH_1\oti \HH_2}$, $\wh\rho \in
\Herm{\HH_1}$, and $\wh\sigma\in \Herm{\HH_2}$ with $\wh\rho, \wh\sigma>0$. Then we have
\begin{equation}
  \label{eq:Comm.Equiv}
[\bbQ,\wh\rho\oti \wh\sigma]=0 \quad \Longleftrightarrow \quad [\bbQ,
\log\wh\rho \oti \mathbf{1}_{\HH_2}] +  [\bbQ,
\mathbf{1}_{\HH_1} \oti \log\wh\sigma ] =0.    
\end{equation}
\end{lemma} 
\begin{proof} We simply note that a Hermitian operator $\bbQ$
  commutes with a Hermitian operator $\bbB>0$ if and only if it
  commutes with its logarithm $\log \bbB$. We apply this to $\bbB=
  \wh\rho\oti \wh\sigma$, for which we have 
\[
\log(\wh\rho\oti \wh\sigma) = \log\wh\rho \otimes \mathbf{1}_{\HH_2} \:
+\: \mathbf{1}_{\HH_1} \otimes \log\wh\sigma.
\]
This gives the desired result. 
 \end{proof}

The above proposition demonstrates that the definition
\eqref{eq:Lind.CompForm} together with
\eqref{eq:CommQSigma} generates DBC Lindbladians. The following
corollary shows that the building blocks 
$\SSS_W$ and $\MM{\beta}{Q}$ from Proposition
\ref{pr:BuildBlock} can be written in the 
compact form \eqref{eq:Lind.CompForm} as well.

\begin{corollary}[Building blocks in compact
  form]\label{co:Biuld.Comp} Consider
  $\RHObeta=\frac1{Z_\beta}\ee^{-\beta H}$ with  $H\in \Herm{\HH_1}$.  
 
(1) Choosing $\HH_2=\C$ and $\bbQ_W=W$ for $W \in \Herm{\HH_1}$ we
have the identity
\[
\SSS_W A  = [W,AW]+[WA,W]= -\big[W,[W,A]\big] =  -\big[
\bbQ_W,[\bbQ_W,A]\big].
\] 
The commutator relation \eqref{eq:Comm.Equiv} for $\bbQ_W$ is simply
$[H,W]=0$.

(2) Choosing $\HH_{2}=\mathbb{C}^{2}$ and $(\omega,Q)\in \mathfrak
E(H)$ we define
\[
\mathbb{Q}_{Q}=\begin{pmatrix}0 & Q^*\\
Q & 0
\end{pmatrix}\quad\wh\sigma_{\beta\omega}=\begin{pmatrix}\ee^{\beta\omega/2} & 0\\
0 & \ee^{-\beta\omega/2}
\end{pmatrix}.
\]
Then, the commutation relation \eqref{eq:Comm.Equiv} holds and we
have
\begin{equation}
\label{eq:BuildingBlockAsTensorProduct}
\begin{aligned}
\MM{\beta}{Q}\,A &
=-\trace_{\HH_{2}} ( \left[\mathbb{Q}_{Q},\left[\mathbb{Q}_{Q},
      A \oti\wh\sigma_{\beta\omega}\right]\right] ) \\
 & =\ee^{\beta\omega/2} ( \left[Q,\rho Q^*\right]+ 
  \left[Q\rho,Q^*\right] ) +\ee^{-\beta\omega/2} 
   ( \left[Q^*,\rho Q\right]+\left[Q^*\rho,Q\right] ) .
\end{aligned}
\end{equation}
\end{corollary} 
\begin{proof} (1) is trivial. For (2) the relation
  \eqref{eq:BuildingBlockAsTensorProduct} follows from a direct
  calculation of the partial trace. In order to check the commutation
  condition \eqref{eq:CommQSigma} we observe 
\[
\left[\mathbb{Q}_{Q},\RHObeta\oti\wh\sigma_{\beta\omega}\right]=\begin{pmatrix}0
  &
  \ee^{-\frac{\beta\omega}{2}}Q^*\RHObeta-\ee^{\frac{\beta\omega}{2}}\RHObeta
  Q^*\\
\ee^{\frac{\beta\omega}{2}}Q\RHObeta-\ee^{-\frac{\beta\omega}{2}}\RHObeta
Q & 0
\end{pmatrix}.
\]
By Lemma \ref{le:Comm.Rel} the relation $\left[Q,H\right]=\omega Q$ is
equivalent to $\ee^{\beta\omega/2} \RHObeta Q=
\ee^{-\beta\omega/2} Q \RHObeta$, so indeed the commutator $ \left[\mathbb{Q}_{Q},\RHObeta\oti\wh\sigma_{\beta\omega}\right]$
vanishes.
\end{proof}

As a last step we want to show that all DBC Lindbladians can be
written in the form \eqref{eq:Lind.CompForm} with a particular choice
of $\wh\sigma$, namely either $\wh\sigma=\RHObeta$ or
$\wh\sigma=\RHObeta^{-1}$. Therefore, we introduce a \emph{partial
  transpose} on $L(\HH_1\oti \HH_2)$ that acts on the $\HH_2$ part and
is associated with a fixed $\wh\sigma= \wh\sigma^*>0$, namely
\begin{equation}
  \label{eq:def.PartTransp}
  \ADJ_{\wh\sigma}: L(\HH_1\oti\HH_2) \to L(\HH_1\oti\HH_2) \ \text{
  is defined via } \ADJ_{\wh{\sigma}} (A\oti B) := A\oti B^{\top_{\wh{\sigma}}} 
\end{equation}
and by linearity on the whole space $L(\HH_1\oti \HH_2)$. To define the
$\wh\sigma$-transpose $B^{\top_{\hat{\sigma}}} $ we write $\wh\sigma =
\sum_{j=1}^J \sigma_j e_j\oti \ol e_j$, where $\set{e_j\in
  \HH_2}{j=1,..,J}$ is an orthonormal basis in $\HH_2$, set
$P_{jk}=e_j\oti \ol e_k$, and set $P_{jk}^{\top_{\wh\sigma}}= P_{kj}$,
which defines $B^{\top_{\hat{\sigma}}} $ by linearity. Clearly,
$\ADJ_{\wh\sigma}\big(\ADJ_{\wh\sigma} \bbQ\big) = \bbQ$ and
$\ADJ_{\wh\sigma} \bbQ$ is Hermitian if and only if $\bbQ$ is
Hermitian.

Moreover, we define a $\wh\sigma$-related operator $\YYY_{\wh\sigma}$
from $L(\HH_1\oti \HH_2)$ into itself via
\begin{equation}
  \label{eq:def.YSigma}
  \YYY_{\wh\sigma} \bbQ :=
  (\bm1_{\HH_1}\oti\wh\sigma^{1/2}) \,\big(\ADJ_{\wh\sigma}\bbQ\big) \,
  (\bm1_{\HH_1}\oti\wh\sigma^{1/2}),
\end{equation}
With $P_{jk}=e_j\oti \ol e_k$ from above, every $\bbQ\in L(\HH_1\oti
\HH_2)$ has the unique representation $\bbQ=\sum_{j,k=1}^J Q_{jk}\oti
P_{jk}$, and we obtain the formulas
\begin{equation}
  \label{eq:bfT.Sigma}
 \ADJ_{\wh\sigma}\bbQ=\sum_{j,k=1}^J Q_{kj}\oti P_{jk} \ \text{ and } \  
\YYY_{\wh\sigma}\bbQ=\sum_{j,k=1}^J \wt Q_{jk}\oti P_{jk} \ \text{
  with } \wt Q_{jk}= (\sigma_j\sigma_k)^{1/2} Q_{kj}. 
\end{equation}
 
The following result indicates how  the partial
$\wh\sigma$-transpose interacts with the commutator of two Hermitian operators. It shows that a representation
\eqref{eq:Lind.CompForm} for $\LLL$ in terms
of $\bbQ$ and $\wh\sigma$ is equivalent to a representation in terms
of $\YYY_{\wh\sigma}\bbQ$ and $\wh\sigma^{-1}$.

\begin{lemma}\label{le:Sigma.wtQ}
Consider $\wh\sigma \in \Herm{\HH_2}$ with
$\wh\sigma>0$. For $\bbQ \in L(\HH_1 \oti \HH_2)$ we
have
\begin{equation}
  \label{eq:wtQ.sigma}
  \ADJ_{\wh\sigma}\big( [ \bbQ, A\oti\wh\sigma]\big) = \big(\bm1_{\HH_1}\oti
  \wh\sigma^{1/2}\big)\, [\YYY_{\wh\sigma}\bbQ, A\oti \wh\sigma^{-1}] \,
\big(\bm1_{\HH_1}\oti \wh\sigma^{1/2}\big) \text{ \ for all }A\in L(\HH_1).
\end{equation}
In particular, we have an equivalence between the commutation relations
\begin{equation}
  \label{eq:Equi.CommR}
  [\bbQ, \RHObeta \oti \wh\sigma]=0 \quad \Longleftrightarrow \quad 
  [\YYY_{\wh\sigma}\bbQ, \RHObeta \oti \wh\sigma^{-1}]=0
\end{equation}
and the dual representations of the compact representation of Lindblad
operators: 
\begin{equation}
  \label{eq:Lind.twice}
  \trace_{\HH_2} \big[ \bbQ,[\bbQ, A\oti \sigma]\big] \ = \ 
 \trace_{\HH_2} \big[ \YYY_{\wh\sigma}\bbQ,[\YYY_{\wh\sigma}\bbQ,
    A\oti  \wh\sigma^{-1}]\big].
\end{equation}
\end{lemma}
\begin{proof}  To simplify the notation we abbreviate $M_{\wh\sigma}:= \bm1_{\HH_1}\oti
  \sigma^{1/2}$.   

  For establishing \eqref{eq:wtQ.sigma}, we can use linearity such
  that it is sufficient to consider the case $\bbQ = Q_{kl}\oti
  P_{kl}$, which gives
\[
\ADJ_{\wh\sigma} \big[ \bbQ, A\oti\wh\sigma\big] = \sigma_l
(Q_{kl}A)\oti P_{lk} - \sigma_k (AQ_{kl})\oti P_{lk}.
\]
Moreover, using $\YYY_{\wh\sigma} \bbQ=(\sigma_k\sigma_l)^{1/2}
Q_{kl}\oti P_{lk}$, we find
\[
\big[\YYY_{\wh\sigma} \bbQ, A\oti\wh\sigma^{-1} \big]=
    (\sigma_l/\sigma_k)^{1/2} (Q_{kl} A)\oti P_{lk} -  
    (\sigma_k/\sigma_l)^{1/2} (A Q_{kl})\oti P_{lk}.
\] 
Since multiplying this expression from the left and from the right by
$M_{\wh\sigma} $ reduces to multiplying by $(\sigma_k\sigma_l)^{1/2}$,
we see that identity \eqref{eq:wtQ.sigma} is established.
 
Clearly, \eqref{eq:Equi.CommR} follows from \eqref{eq:wtQ.sigma} by
choosing $A=\RHObeta$ and using that $M_{\wh\sigma}$ is invertible.

Identity \eqref{eq:Lind.twice} follows by recalling the relation
\eqref{eq:calL.bbQ.sigma}, which gives
\[
-\trace_{\HH_2} \big[ \bbQ,[\bbQ, A\oti \sigma]\big]= \sum\nolimits_{k,l=1}^J
\sigma_l\big( [Q_{kl} A, Q_{kl}^*]+[Q_{kl},AQ_{kl}^*] \big). 
\]
Applying the same formula but with $\bbQ$ and $\wh\sigma$ replaced by
$\YYY_{\wh\sigma} \bbQ$ and $\wh\sigma^{-1}$ we simply have to replace $Q_{kl}
$ by $\wt Q_{kl}=(\sigma_k\sigma_l)^{1/2} Q_{lk}$ and the eigenvalues 
$\sigma_l$ by  $1/\sigma_l$. We then find the same result,
and \eqref{eq:Lind.twice} is proved.
\end{proof}

We now come to our representation of DBC Lindbladians with the choice
$\wh\sigma=\RHObeta$ or $\wh\sigma=\RHObeta^{-1}$, 
which are both useful and have a natural interpretation. In the first
case, we can use the fact that for all $\rho\in \RRR\subset L(\HH)$  the
tensor product $\rho\oti\RHObeta$ is again a density matrix, but now
on the Hilbert space $\HH\oti \HH$. In the second case the matrix
$\rho\oti \RHObeta^{-1}$ can be seen as a non-commutative counterpart
of the relative density $\varrho(x)=u(x)/U^\eq(x)$ in the
Fokker-Planck equation or of the relative density
$(p_n/w^\eq_n)_{n=1,..,N}$ for discrete Markov processes, see Section
\ref{su:ClassMP}.  Note also that the two commutator relations
\[
[\bbQ, \RHObeta\oti\RHObeta]=0 \quad \text{and} \quad [\wt\bbQ,
\RHObeta \oti \RHObeta^{-1}]=0
\]
look quite different, since $\RHObeta\oti \RHObeta$ has the
eigenvalues $\frac1{Z_\beta}\ee^{-\beta(\eps_j+\eps_k)}$ while
$\RHObeta\oti \RHObeta^{-1}$  has the
eigenvalues $\frac1{Z_\beta}\ee^{-\beta(\eps_j-\eps_k)}$. So the
latter appears closer to the relevant eigenpairs $(\omega,Q)\in
\mathfrak E(H)$. However, we will see in the following theorem that there is a one-to-one correspondence
between all possible $\bbQ$ and $\wt\bbQ$. Its proof is
based on the previous lemma.

\begin{theorem}[Compact representation of $\LLL$ with DBC]
\label{th:CompRepDBC}
Let $\LLL$ be a DBC Lindblad operator with respect to $\RHObeta \in
L(\HH)$. Then, there exists $\bbQ\in \Herm{ \HH \oti
  \HH_2} $ with $\HH_2 = \HH$ satisfying the commutator relation
$[\bbQ,\RHObeta \oti \RHObeta ]=0$ such that the representation
\[
\LLL \rho =-\trace_{\HH_{2}} \Big( \big[\mathbb{Q},
[\mathbb{Q},\rho\oti \RHObeta ]\big] \Big) 
\]
holds. Moreover, choosing $\wt\bbQ =\YYY_{\RHObeta} \bbQ$ as in Lemma
\ref{le:Sigma.wtQ}, we have the alternative representation
\[
\LLL \rho = - \trace_{\HH_2} \Big( \big[ \wt\bbQ, 
[\wt\bbQ,\rho\oti \RHObeta^{-1} ]\big] \Big)  .
\]
\end{theorem}
\begin{proof}
By \cite{KFGV77QDBK} every DBC Lindblad operator can be written in the form
\[
\LLL ( \rho ) =\sum_{ij,mn}M_{ij,mn} (
\left[P_{ij}\rho,P_{mn}^*\right]+\left[P_{ij},\rho P_{mn}^*\right] )  
\]
with $P_{ij}=h_{i}\oti\overline{h}_{j}$, where $h_{i}$ are the
eigenvectors of 
$\RHObeta$, and $M_{ij,mn}$ satisfy
\begin{align*}
\text{(i) }& \overline{M}_{mn,ij}=M_{ij,mn}, \qquad \text{(ii) } 
\eps_{j}-\eps_{i}\neq\eps_{n}-\eps_{m}\
\Longrightarrow \  M_{ij,mn}=0, \\
\text{(iii) }& M_{nm,ji}=\ee^{-\beta\omega}M_{ij,mn} \text{ with } 
 \omega=\eps_{j}-\eps_{i}=\eps_{n}-\eps_{m}.
\end{align*}
We construct the Hermitian operator $\bbQ$ in the form 
$\mathbb{Q}=\sum_{i,j,m,n}A_{ij,kl}P_{ij}\oti P_{kl}^*$. Hence, 
\begin{align*} 
[\mathbb{Q},\RHObeta\oti\RHObeta]=0 \quad &\Longleftrightarrow
\quad \Big( A_{ij,kl}=0 \text{ whenever }
\eps_{j}{-}\eps_{i}\neq\eps_{l}{-}\eps_{k} \Big),  
\\
\mathbb{Q}^*=\mathbb{Q} \quad &\Longleftrightarrow
\quad  \overline{A}_{ij,kl}=A_{ji,lk}.
\end{align*}
To this end we define
\[
\tilde{M}_{ij,mn}=Z_{\beta}^{2}\, \ee^{\beta(\eps_{i}+\eps_{m})/2}M_{ij,mn}.
\] 
Then $\tilde{M}_{ij,mn}=\tilde{M}_{nm,ji}=\overline{\tilde{M}}_{ji,mn}$.
This symmetry property remains true for all powers of $\tilde{M}$
and thus for $\tilde{M}^{\frac{1}{2}}$ as well. Define $A_{ij,kl}=\rme^{\beta\frac{\eps_{k}-\eps_{i}}{2}} ( \tilde{M}^{\frac{1}{2}} ) _{ij,kl}.$
Then,  
\begin{align*}
A_{ji,lk} & =\rme^{\beta\frac{\eps_{l}-\eps_{j}}{2}} (
\tilde{M}^{\frac{1}{2}} ) _{ji,lk}
%\\ & 
  =\rme^{\beta\frac{\eps_{l}-\eps_{j}}{2}}\overline{ ( \tilde{M}^{\frac{1}{2}} ) _{ij,kl}}\\
 & =\rme^{\beta\frac{\eps_{l}-\eps_{k}}{2}-\beta\frac{\eps_{j}-\eps_{i}}{2}}\overline{A_{ij,kl}}=\rme^{\beta\frac{\omega-\omega}{2}}\overline{A_{ij,kl}}
\end{align*}
Thus the corresponding $\mathbb{Q}$ is Hermitian and $A_{ij,kl}=0\quad\text{if }\eps_{j}-\eps_{i}\neq\eps_{l}-\eps_{k}$
follows from condition (2) on $M_{ij,mn}$. Finally 
\begin{align*}
&\frac{1}{Z_{\beta}^{2}}\sum_{k,l}A_{ij,kl}\overline{A_{mn,kl}}\rme^{-\beta\eps_{k}}  =\frac{1}{Z_{\beta}^{2}}\sum_{k,l}\rme^{\beta\frac{\eps_{k}-\eps_{i}}{2}} ( \tilde{M}^{\frac{1}{2}} ) _{ij,kl}\cdot\rme^{\beta\frac{\eps_{k}-\eps_{m}}{2}}\overline{ ( \tilde{M}^{\frac{1}{2}} ) _{mn,kl}}\rme^{-\beta\eps_{k}}\\
 &
 =\rme^{-\beta\frac{\eps_{i}+\eps_{m}}{2}}\frac{1}{Z_{\beta}^{2}}\sum_{k,l}
 ( \tilde{M}^{\frac{1}{2}} ) _{ij,kl} ( \tilde{M}^{\frac{1}{2}} ) _{kl,mn} 
%\\ &
 =\rme^{-\beta\frac{\eps_{i}+\eps_{m}}{2}}\frac{1}{Z_{\beta}^{2}}\tilde{M}_{ij,mn}=H_{ij,mn}
\end{align*}
which means that
\[
\LLL\rho=\sum_{ij,mn}M_{ij,mn} \big(
[P_{ij}\rho,P_{mn}^*]+[P_{ij},\rho P_{mn}^*] \big) =
-\trace_{\HH_{2}} \Big( \big[\mathbb{Q},
[\mathbb{Q},\rho\oti\RHObeta] \big]  \Big)  .
\]
This establishes the first representation based on $\bbQ$ and
$\RHObeta$. The second representation involving $\wt\bbQ_{\RHObeta}$
and $\RHObeta^{-1}$ follows simply by
applying \eqref{eq:Lind.twice} to the case $\sigma=\RHObeta$.
 \end{proof}

\subsection{Examples of Lindblad operators and equations}
\label{su:Lind.Exa}

Here we give two elementary examples to highlight the structures and
to come back to them in later sections. We refer to
\cite{BaNaTh08AQSG1,BauNar08AQSG2,BauNar12SSSQ,AlbLia14SCQL} for general
discussions of the dynamics of Lindblad equations with or without
DBC. Even without explicitly mentioning the DBC, it was already
observed in \cite[Thm.\,4.4]{Davi74MME} that under the additional
assumption that the eigenvalues of $H$ are all simple, the diagonal
elements $\rho_{jj}(t)$, $j=1,...,N$, of $\rho(t)$ with respect to the
eigenbasis of $H$ evolve according to a 
classical Markov process, see also \cite[Eqn.\,(21)]{Alic76DBCN}.

\subsubsection{The Bloch sphere for the case $N=2$}
\label{suu:Bloch}

For the case $\HH=\C^2$ and $H=\mathrm{diag}(\eps_1,\eps_2)$ with
$\eps_1\neq \eps_2$ we characterize all DBC Lindbladians $\LLL$ with
respect to $\RHObeta$. For this we use the Pauli matrices
\[
\sigma_{1}=\begin{pmatrix}0 & 1\\
1 & 0
\end{pmatrix},\quad\sigma_{2}=\begin{pmatrix}0 & -\ii\\
\ii & 0
\end{pmatrix},\quad\sigma_{3}=\begin{pmatrix}1 & 0\\
0 & -1
\end{pmatrix}.
\]
and the operators $\sigma_\pm = \frac{1}{2}(\sigma_1 \pm \sigma_2)$.  
Then, all $\LLL$ satisfying the DBC have the form
\begin{align}
  \LLL_{d} ( \rho )  & =\frac{\gamma}{2}\Big( \ee^{-\beta
      \eps_{1}} \big( \left[\sigma_{+}\rho,\sigma_{-}\right] {+}
    \left[\sigma_{+},\rho\sigma_{-}\right] \big)  
  +\ee^{- \beta \eps_{2}} \big( \left[\sigma_{-}\rho,\sigma_{+}\right]
  {+} \left[\sigma_{-},\rho\sigma_{+}\right] \big) \Big) \nonumber \\
  & \quad+\frac{\delta}{2} \big( \left[\sigma_{3}\rho,\sigma_{3}\right]
  {+} \left[\sigma_{3},\rho\sigma_{3}\right] \big), 
  \label{eq:DissipativeLindbladGenerator}
\end{align}
which are simply two building blocks with $Q=\sigma_+$,
$Q^*=\sigma_-$, and $W=\sigma_3$.

It is more convenient to write the above generator in terms of
real-valued Bloch coordinates $\mathbf{a}\in \mathbb{R}^3$ via $\rho (
\mathbf{a} ) = \frac{1}{2} \mathrm{id}+\frac{1}{2} \sum_{i=1}^3 a_i
\cdot \sigma_i$. The positivity $\rho \geq 0$ is equivalent to $\left|
  \mathbf{a} \right| \leq 1$. The Lindblad equation
$\dot{\rho}=\LLL_d ( \rho ) $ reads $\dot{\mathbf{a}}=R
\mathbf{a} + \mathbf{k}$ in Bloch coordinates with
\begin{equation}
R=\begin{pmatrix}- ( \gamma {+}2\delta )  & 0 & 0\\
0 & - ( \gamma{+}2\delta )  & 0\\
0 & 0 & -2\gamma
\end{pmatrix},\quad\mathbf{k}=\begin{pmatrix}0\\
0\\
2\gamma\overline{w}
\end{pmatrix}\label{eq:BlochRepresentationDetailedBalance}
\end{equation}
and with $\overline{w_\beta}= \ee^{-\beta \eps_1}- \ee^{-\beta
  \eps_2}$. This is the dissipative part of the well-known
phenomenological Bloch equations. The longitudinal and transverse
relaxation times $T_{1}$ and $T_{2}$ are given by
$T_{1}=\frac{1}{2\gamma}$ and $T_{2}=\frac{1}{\gamma+2\delta}.$ They
satisfy the inequality $T_{1}\geq\frac{1}{2}T_{2}$.

\subsubsection{A nontrivial case}
\label{suu:NontrivialCase}

In this example we give a case where an eigenvalue $\omega \neq
0$  of $A\mapsto
[A,H]$ is not simple, which allows for a nontrivial coupling
between four energy levels.

In $\HH=\C^4$ we choose $H=\sum_{i=1}^{4}\eps_{i}h_{i}$ with
$\eps_{4} = 10$, $\eps_{3} = 9$, $\eps_{2} = 2$, $\eps_{1} = 1$, and
$Q=h_{1}\oti\overline{h}_{2}+h_{3}\oti\overline{h}_{4}$. Using
$\eps_4 {-} \eps_3 = \eps_2 {-}\eps_1=1$ we have $(1,Q)\in \mathfrak
E(H)$.  Then, by Proposition \ref{pr:BuildBlock}(b) we see that
\begin{equation}
\dot{\rho}=\MM{\beta}{Q}\rho=\ee^{\frac{\beta}{2}} (
\left[QA,Q^*\right]+\left[Q,AQ^*\right] ) +\ee^{-\frac{\beta}{2}} (
\left[Q^*A,Q\right]+\left[Q^*,AQ\right]
) \label{eq:ExampleDegenerateEnergyDifferences} 
\end{equation}
satisfies the DBC with respect to $\RHObeta$.
Let us rewrite the above equation in coordinates. The diagonal elements
form a Markov chain
\begin{align*}
\dot{\rho}_{11} & =-\ee^{-\frac{\beta}{2}}\rho_{11}+\ee^{\frac{\beta}{2}}\rho_{22} & \quad\dot{\rho}_{33} & =-\ee^{-\frac{\beta}{2}}\rho_{33}+\ee^{\frac{\beta}{2}}\rho_{44}\\
\dot{\rho}_{22} & =+\ee^{-\frac{\beta}{2}}\rho_{11}-\ee^{\frac{\beta}{2}}\rho_{22} & \quad\dot{\rho}_{44} & =+\ee^{-\frac{\beta}{2}}\rho_{33}-\ee^{\frac{\beta}{2}}\rho_{44}
\end{align*}
and the evolution of the off-diagonal elements is given by 
\[
\dot{\rho}_{13}=-\ee^{-\frac{\beta}{2}}\rho_{13}+\ee^{\frac{\beta}{2}}\rho_{24}\quad\dot{\rho}_{24}=+\ee^{-\frac{\beta}{2}}\rho_{13}-\ee^{\frac{\beta}{2}}\rho_{24}
\]
and $\dot{\rho}_{kl}=-\cosh\frac{\beta}{2}\,\rho_{kl}$ for $ ( k,l )
\notin\left\{ ( 1,3 ) , ( 3,1 ) , ( 2,4 ) , ( 4,2 ) \right\} .$ This
example shows, that non-diagonal elements (here $\rho_{13}$ and
$\rho_{24}$) can couple, if the energy differences are the same.  Note
that $\rho_{11}$ and $\rho_{22}$ are decoupled from $\rho_{33}$ and
$\rho_{44}$. Thus $\RHObeta$ is not the only equilibrium of
\eqref{eq:ExampleDegenerateEnergyDifferences}. This is also the reason
why $\rho_{13}$ and $\rho_{24}$ do not decay to $0$, contrary to the
other off-diagonal elements. 

In the presence of symmetries or conserved quantities even more
complicated situations are possible, see \cite{AlbLia14SCQL}.

\section{An entropic gradient structure for the Lindblad equation}
\label{se:EntrGSQME}

\subsection{Entropic gradient structures for classical Markov processes}
\label{su:ClassMP}
The entropic gradient structure for master equations for classical
Markov processes goes back to the seminal work
\cite{JoKiOt97FEFP,JoKiOt98VFFP}, where the Fokker-Planck equation
\[
\dot u = \div \Big( a(x) \big(\nabla u + u \nabla V(x)\big) \Big)
\]
for the probability density $u(t,x)\geq 0$ was written as a gradient
system with respect to the Wasserstein 
distance. Here $a(x)\in \R^{d\ti d}$ is a symmetric and positive
definite diffusion matrix. The gradient structure has the form 
\[
\dot u = -\bbK_\rmW(u) \,\rmD\calF(u),
\]
where $\calF$ is the free energy (or relative entropy with respect to
the equilibrium density $U^\eq(x)=\ee^{-V(x)}$) and
$\bbK_\rmW$ is the Onsager operator associated with the Wasserstein
distance, namely 
\begin{align*}
&\calF(u) = \int_{\R^d} \Big(u(x) \log u(x) + V(x) u(x) \Big)\dd x =
\int_{\R^d} u(x) \log\big( u(x)/U^\eq(x)\big) \dd x ,\\
&\bbK_\rmW(u) \xi = - \div\big( u\, a(x) \,\nabla \xi \big). 
\end{align*}  

A related gradient structure for time-continuous Markov processes on a
discrete state space $\{1,\ldots,N\}$ was found independently by the
three groups \cite{Maas11GFEF,ErbMaa12RCFM}, \cite{CHLZ12FPEF}, and
\cite{Miel11GSRD,Miel13GCRE}. In this case the Kolmogorov forward
equation  for the probability vector
$p(t)\in \set{(p_1,\ldots,p_N)\in [0,1]^N}{\sum_{n=1}^N p_n=1 }$ is
the linear system
\[
\dot p = L p , \quad \text{where } L_{nm}\geq 0 \text{ for }n\neq m
\text{ and } L^\top(1,1,..,1)^\top=0.
\]
The detailed balance condition for $L$ and the equilibrium $w^\eq$ reads 
\[
Lw^\eq =0 \text{ with }w^\eq_n>0 \quad \text{and} 
\quad \kappa_{nm}:=L_{nm}w^\eq_m=L_{mn} w^\eq_n \text{ for all }n,m\in \{1,..,N\}.
\] 
The entropic gradient structure is defined in terms of the relative
entropy $\calE(p)=\calH(p|w^\eq)$ and the Onsager operator
$\bbK_\rmM(p)$ with 
\[
\calE(p)=\sum_{n=1}^N p_n\log\big( p_n/w^\eq_n\big)
\text{ and } 
\bbK_\rmM(p)= \sum_{m>n} \kappa_{nm} \Lambda\big(\frac{p_n}{w_n^\eq},
\frac{p_m}{w^\eq_m} \big) \Big(e_n{-}e_m\Big)\oti \Big(e_n{-}e_m\Big),
\]
where $\Lambda(a,b)\geq 0$ denotes the \emph{logarithmic mean of $a$
  and $b$}:
\begin{equation}
  \label{eq:LogMean}
  \Lambda(a,b) = \int_0^1 a^s \,b^{1-s} \dd s = \frac{a\;-\; b}{\log
    a{-}\log b} . 
\end{equation}
Note that using $\rmD\calE(p)=\big(\log p_n - \log
w^\eq_n\big)_{n=1,..,N}$, the relation $\Lambda(a,b)\big(\log a - \log
b\big) = a- b$, and the detailed balance condition easily yield the
identity $Lp = - \bbK_\rmM(p)\rmD\calE(p)$.

\subsection{The Kubo-Mori operator $\CCC_\rho$ and the generalization
  $\DDD^\alpha_\rho$} 
 \label{su:KuboMori}

 The development of an analogous gradient structure for the
 dissipative part of the Lindblad equation was
 less successful. The attempts in
 \cite{Otti10NTQM,Otti11GTDQ,Miel13DQMU,Miel15TCQM} produced nonlinear
 terms, unless the Hamiltonian $H$ is a multiple of $\bm1_\HH$ (as in
 \cite{CarMaa14A2WM} or more generally only the building blocks $\SSS_W$
 in Proposition \ref{pr:BuildBlock} are used).  All of these works 
 involve the Kubo-Mori operator $\CCC_\rho:L(\HH) \to L(\HH)$ as a
 generalization of the 
 multiplication with $u$ in the Fokker-Planck equation and the
 logarithmic mean $\Lambda(\frac{p_n}{w^\eq_n},
 \frac{p_m}{w^\eq_m})$. It is defined via
\[
\CCC_\rho A:= \int_0^1 \rho^s\,A\,\rho^{1-s} \dd s \ = \ 
\sum_{n,m=1}^N \Lambda(r_n,r_m)\,\SP{\psi_n}{A\psi_m}\,
\dyad{\psi_n}{\psi}_m ,
\]
if $\rho$ is given by  \eqref{eq:RepresRho}.
\begin{comment}
Moreover, using $\rho\in \RRR$ we have the formulas $(\CCC_\rho A)^* =
\CCC_\rho (A^*)$ and $\SPP{A}{\CCC_\rho A}\geq 0$. We refer to
\cite[Sec.\,21.4.1]{Miel13DQMU} for more details and background on
$\CCC_\rho$, and in particular to \cite[Prop.\,21.2]{Miel13DQMU} for a
continuity result of the mapping $\RRR\ni \rho \mapsto \CCC_\rho$,
which holds even in the infinite dimensional setting.
\end{comment}

One major property of $\CCC_\rho$ is that it satisfies the analog of
the identities 
\begin{equation}
  \label{eq:LinRHS}
   u \nabla \log\big( u/\ee^{-V}\big)  = \nabla u + u \nabla V \quad
 \text{and} \quad \Lambda(a,b)(\log  a{-}\log b) = a - b
\end{equation}
for the classical Markov setting. Note that the right-hand sides are
linear in $u$ and $(a,b)$, respectively. 
For all $Q\in L(\HH)$ the operator $\CCC_\rho$ satisfies a similar ``miracle
identity'', namely  
\begin{equation}
  \label{eq:CCC-mira}
  \CCC_\rho[Q,\log \rho]= [Q,\rho],
\end{equation}
see \cite{Otti10NTQM,Otti11GTDQ,Miel13DQMU,Miel15TCQM}. 
We will provide a proof of a more general version of this identity in
Proposition \ref{pr:genMira}.  

This relation works well (see \cite{CarMaa14A2WM}) 
if we are using the total entropy 
$\calS_0(\rho)=- \trace(\rho\log \rho)$ which has the derivative 
$\rmD\calS(\rho)=-\log \rho$ (up to an identity which is irrelevant
since $\trace \rho =1$). 
However, for relative entropies of the form
$\calS_\beta(\rho)=-\trace(\beta H \rho + \rho \log\rho)$ we have 
\[
\rmD\calS_\beta( \rho)=-\beta H - \log \rho \quad \leadsto 
\quad \CCC_\rho\,[Q,\rmD\calS_\beta( \rho)] =-[Q,\rho] -
\CCC_\rho[Q,H].
\]
Thus, the right-hand side is no longer linear, unless $Q$ commutes
with $H$. The Fokker--Planck equation studied in \cite{CarMaa14A2WM}
has $H=0$ and hence falls into this class, i.e. the Fokker-Planck
equation is indeed a linear Lindblad equation. However, the models
studied in \cite{Otti10NTQM,Otti11GTDQ,Miel13DQMU,Miel15TCQM} include
the nonsmooth term  $\CCC_\rho[Q,H]$, which is continuous but not
H\"older continuous, so the existence theory developed in
\cite[Sec.\,21.6]{Miel13DQMU} is nontrivial and uniqueness of
solutions couldn't be established. 
\medskip

We now show that it is possible to use variants of $\CCC_\rho$ such
that for $(\omega,Q) \in \mathfrak E(H)$ we obtain a suitable
counterpart of \eqref{eq:LinRHS}. Indeed we will be able to show that
all DBC Lindbladians can be written in terms of these variants of
$\CCC_\rho$. The variant of $\CCC_\rho$ we are using is defined in
terms of the tilted operator $\DDD^\alpha_\rho$, where $\alpha \in \R$
will be related to an energy difference:
\begin{equation}
\label{eq:Def-D}
\DDD^\alpha_\rho A := \ee^{-\alpha/2}\int_0^1 \ee^{s \alpha} \rho^s A
\,\rho^{1-s} \dd s \ = \ \sum_{n,k=1}^N \Lambda\big( \ee^{\alpha/2} r_n\,,\,
\ee^{-\alpha/2}r_k\big) \: \SP{\psi_n}{A|\psi_k}\: \psi_n\oti \ol\psi_k,
\end{equation}
if $\rho$ is given by  \eqref{eq:RepresRho}. Again the logarithmic
mean $\Lambda(a,b)$ from \eqref{eq:LogMean} is involved, but now
weighted with $\ee^{\pm \alpha/2}$.

The generalized miracle identity is given in the following result
\eqref{eq:genMiracle},
which again shows that applying $\DDD^\alpha_\rho$ to a commutator
with $\log \rho$ plus a suitable correction provides a linear
expression, i.e.\ the nonlinearities involved in $\log \rho$ and
$\DDD^\alpha_\rho$ cancel each other. 

\begin{proposition}\label{pr:genMira}
For all $\alpha\in\R$, $A, Q\in L(\HH ) $, and $\rho \in \RRR_N$ we
have the identities
\begin{align}
&\label{eq:DDD.adj} 
\big(\DDD^\alpha_\rho\big)^*= \DDD^{\alpha}_\rho \quad \text{and}
\quad \big(\DDD^\alpha_\rho A^*\big)^* = \DDD^{-\alpha}_\rho A, \\
&\label{eq:DDD.pos} \SPP{A}{\DDD^\alpha_\rho A} \geq 0, 
\\
\label{eq:genMiracle}
 & \DDD^\alpha_\rho\big( [Q,\log \rho]-\alpha Q\big) = \ee^{-\alpha/2}Q\rho -
  \ee^{\alpha/2} \rho Q. 
\end{align}
\end{proposition} 
From the proof it is clear, that the generalized miracle identity
\eqref{eq:genMiracle} also holds for $\alpha \in \C$, but our use will
be restricted to real-valued energy levels. \smallskip 
\\
\begin{proof} The relations in \eqref{eq:DDD.adj} follow directly from
  the definition.  For \eqref{eq:DDD.pos} we use
\begin{align*}
  \SPP{A}{\DDD^\alpha_\rho A}&= \ee^{-\alpha/2}\int_0^1 \ee^{s\alpha}
  \SPP{A}{\rho^s A \rho^{1-s} } \dd s\\& = \int_0^1
  \ee^{\alpha(s-1/2)} \trace\big( (\rho^{s/2}A\rho^{s/2})^*
  \rho^{s/2}A\rho^{s/2} \rho^{1-2s}\big) \dd s \geq 0,
\end{align*}
since the integrand is non-negative for all $s\in [0,1]$. 

For \eqref{eq:genMiracle}, we generalize the simple proof of
\eqref{eq:CCC-mira} from \cite[Prop.\,21.1]{Miel13DQMU}, write
$\Lambda = \log \rho $, and use the fact that the integrand defining
$\DDD^\alpha_\rho$ can be written as a total derivative with respect
to $s\in [0,1]$:
\begin{align*}
&\DDD^\alpha_\rho\big([ Q,\log \rho]-\alpha Q\big) 
 = \ee^{-\alpha/2} \int_0^1 \ee^{\alpha s} \ee^{s\Lambda} \big(
 Q\Lambda- \Lambda Q - \alpha \Lambda\big) \ee^{(1{-}s) \Lambda} \dd s \\
& = - \ee^{-\alpha/2} \int_0^1 \Big( \ee^{\alpha s} \ee^{s\Lambda}
(\Lambda{+}\alpha I)Q\ee^{(1{-}s) \Lambda}  +\ee^{\alpha s}
\ee^{s\Lambda} Q\big({-}\Lambda\big)  \ee^{(1{-}s) \Lambda} \Big)\dd s \\
&=- \ee^{-\alpha/2} \int_0^1 \frac{\rmd}{\rmd s} \Big(\ee^{(\Lambda{+}\alpha I)s} 
   Q\ee^{(1-s)\Lambda} \Big) \dd s \ = \ \ee^{-\alpha/2} \big(Q\ee^\Lambda -
\ee^{\Lambda{+}\alpha I}Q\big)\\ 
&=  \ee^{-\alpha/2} Q\rho- \ee^{\alpha/2} \rho Q  .
\end{align*}
This is the desired result.
\end{proof}
The following result follows immediately from the above
proposition by setting $\alpha=-\beta\omega$. It will be the basis for
our construction of the entropic gradient structure.

\begin{corollary}\label{co:DDD-lin}
Assume $\beta>0$ and that $(\omega,H)\in \mathfrak E(H)$, that is
$[Q,H]=\omega Q$, then  
\begin{equation}
  \label{eq:DDlinRHS}
  \DDD^{-\beta\omega}_\rho \,\big[\ Q\,,\,\log \rho+\beta H\;\big] 
  \ = \ \ee^{\beta\omega/2}Q\rho -  \ee^{-\beta\omega/2} \rho Q. 
\end{equation}
\end{corollary} 

The next lemma shows that $\DDD_\rho^\alpha$ appears naturally if we
tensorize $\rho$ with a diagonal matrix, and thus connects our
construction with that in Section \ref{su:Lind.CompForm}.

\begin{lemma}\label{le:CCCab} 
For all $\rho\in \RRR$ and all $\alpha\in \R$ we have  
\[
\CCC_{\block{\ee^{\alpha/2} \rho}00{\ee^{-\alpha/2}\rho}}\block ABCD =
\block{\ee^{\alpha/2} \CCC_\rho A}{\DDD^{\alpha}_\rho B}
 {\DDD^{-\alpha}_\rho C}{ \ee^{-\alpha/2}\CCC_\rho D}. 
\]
\end{lemma} 
\begin{proof} This follows simply by using the identity (for $a,b>0$)
\[
\block{a\rho}00{b\rho}^s=
  \block{a^s\rho^s}00{b^s\rho^s} 
\] 
and the definitions of $\CCC_{\binom{\rho_1\ 0\;}{\;0\ \rho_2}}$ and
$\DDD_\rho^\alpha$ from above. 
\end{proof}

As was pointed out in Section \ref{su:Lind.CompForm} every Lindblad
operator can be written as the sum of partial traces of double
commutators on a larger tensor product space. This enlarged space
has the advantage, that the miracle identity \eqref{eq:genMiracle}
becomes very elegant and more transparent, when taking the
original miracle identity for granted.

\begin{theorem}[Generalized miracle identity]\label{th:MiracleIdentity}
  Consider $\bbQ\in \Herm{\HH_1\oti \HH_2}$, $\wh\rho\in \Herm{\HH_1}$ and 
  $\wh\sigma \in \Herm{\HH_2}$ with $\wh\rho, \wh\sigma>0$ satisfying the
  commutator relation 
  \eqref{eq:CommQSigma}, i.e. $[\bbQ , \RHObeta \oti \wh\sigma ]
  = 0$. Then, for all $\rho\in \RRR$ we have the identity  
\begin{equation} 
\label{eq:tensorizedMiracle}
\CCC_{\rho\otimes\wh\sigma} \big[ \mathbb{Q}, ( \log\rho {-}
\log\RHObeta ) \oti\mathbf{1}_{\HH_{2}} \big] = \left[\mathbb{Q},\rho
  \oti \wh\sigma\right]. 
\end{equation}
\end{theorem}
\begin{proof}
Using Lemma \ref{le:CommQSigma}
in ``$\overset{\star}=$'' below
we obtain the following chain of identities:
\begin{align*}
&\CCC_{\rho\oti\wh\sigma} [\mathbb{Q}, ( \log\rho{-}\log\wh\rho)
  \oti \mathbf{1}_{\HH_{2}} ] =
  \CCC_{\rho \oti \wh\sigma} [\mathbb{Q},  (  \log\rho  )  \oti
  \mathbf{1}_{\HH_{2}} ] -\CCC_{\rho\oti\wh\sigma}[\bbQ,(\log\wh\rho) \oti
  \mathbf{1}_{\HH_{2}} ]
\\ 
& \overset{\star}= \CCC_{\rho \oti
  \wh\sigma}\left[\mathbb{Q},  (  \log\rho  )  \oti
  \mathbf{1}_{\HH_{2}} ] +\CCC_{\rho\oti\wh\sigma}[\bbQ, \mathbf{1}_{\HH_{1}}\oti  (  \log\wh\sigma  )  \right]\\
 & =\CCC_{\rho\otimes\wh\sigma}\big[\mathbb{Q}, (  \log\rho  )  \oti
  \mathbf{1}_{\HH_{2}} {+}  \mathbf{1}_{\HH_{1}}\oti  (  \log\wh\sigma  ) \big]
=\CCC_{\rho\otimes\wh\sigma}\left[\mathbb{Q},\log ( \rho\oti\wh\sigma
  ) \right]
\overset{\circ}= \left[\mathbb{Q},\rho\oti\wh\sigma\right],
\end{align*}
 where ``$\overset{\circ}=$'' uses the classical miracle identity
\eqref{eq:CCC-mira}. 
\end{proof}

We note that relation \eqref{eq:DDlinRHS} in Corollary
\ref{co:DDD-lin} is a direct consequence of Theorem
\ref{th:MiracleIdentity} by setting
\[
\mathbb{Q}= \block{0}{Q^\ast}{Q}{0} \quad \wh\sigma = 
\block{\ee^{\beta \omega/2}}00{\ee^{-\beta \omega/2}}.
\]

\subsection{Dissipation potential and Onsager operator}
\label{su:DissPot}

We now complete the task of writing the dissipative part $\LLL$
of any DBC Lindblad operator with respect to $\RHObeta$
as a gradient of the relative entropy, namely
\[
\calF(\rho)=\calH(\rho|\RHObeta):= \trace \Big( \rho\, \big( \log
\rho - \log \RHObeta\big) \Big) = \trace\big( \rho\log \rho + \rho\,\beta
H)+\log Z_\beta.
\]
The aim is to construct an Onsager operator $\bbK(\rho)$ such that 
\[
\LLL \rho = -\bbK(\rho) \rmD \calF(\rho) = -\bbK(\rho)\big( \log\rho +
\beta H\big) .
\]

The symmetric and positive definite Onsager operator $\bbK(\rho)$ is
most easily defined in terms of a non-negative and quadratic dual
dissipation potential 
\[
\calR^*(\rho,\xi)= \frac12\SPP\xi{\bbK(\rho)\xi}.
\]
Such a structure is conveniently written down in the compact 
tensor product formulation for Lindblad operators $\LLL$ as developed in Section
\ref{su:Lind.CompForm}, namely 
\begin{equation}
  \label{eq:Lind.Tensor}
  \LLL\rho= - \trace_{\HH_2}\Big( \big[ \bbQ, [\bbQ, \rho\oti\wh
\sigma]\big] \Big). 
\end{equation}
As a corollary we will also obtain the corresponding gradient
structures for the building blocks of Proposition
\ref{pr:BuildBlock}, namely 
\[
\LLL\rho= \sum_{j=1}^J  \MM{\beta}{
  Q_j} \,\rho \quad\text{where }(\omega_j,Q_j)\in \mathfrak E(H).
\] 
Indeed, each building block can be expressed in tensor form with
$\wh\sigma_j=1$ or $\wh\sigma_n\in \R^{2\ti 2}$. If we construct a
suitable $\KKK_j(\rho)$ for each of the building blocks, we can use
the additivity principle for Onsager operators, i.e.\ the sum
$\bbK(\rho):= \sum_{j=1}^J \KKK_j(\rho)$ is the desired total Onsager
operator.  
   
The tensorial representation \eqref{eq:Lind.Tensor} for  DBC
Lindblad operators with respect to $\RHObeta$ and the
tensorial miracle identity \eqref{eq:tensorizedMiracle} lead us to
the following general result. 

\begin{proposition}[Onsager operators and dissipation potentials]
  \label{pr:TensorGS} For given $\mathbb{Q}$, $\RHObeta$, and
  $\wh\sigma$ satisfying the commutation relation
  $\left[\mathbb{Q},\RHObeta \oti \wh\sigma\right]=0$, we define
  the Onsager operator $\bbK (\rho ):\Herm{\HH_1} \rightarrow \Herm{\HH_1}$ and the
  dual dissipation potential
  $\mathcal{R}^* (\rho, \cdot) :\Herm{\HH_1}\rightarrow\mathbb{R}$ via
\begin{align*}
\bbK ( \rho ) \xi & :=\trace_{\HH_{2}} \Big( \big[\mathbb{Q},\CCC_{\rho \oti \wh\sigma}[\mathbb{Q},\xi \oti \mathbf{1}_{\HH_{2}}]\big] \Big) \\
\mathcal{R}^* ( \rho,\xi )  & :=\frac12\Big\langle \big[
\mathbb{Q},\xi \oti \mathbf{1}_{\HH_{2}}\big],\CCC_{\rho \oti
  \wh\sigma}\big[\mathbb{Q},\xi \oti
\mathbf{1}_{\HH_{2}}\big]\Big\rangle  .
\end{align*}
Then,
\begin{equation}
  \bbK ( \rho )  ( \log\rho-\log\RHObeta ) 
   =\trace_{\HH_{2}}
   ( \left[\mathbb{Q},\left[\mathbb{Q},
        \rho \oti \wh\sigma\right]\right] ) 
  \label{eq:TensorGS} 
\end{equation}
as well as
\[
\bbK ( \rho ) =\bbK ( \rho ) ^*\geq0, \quad\bbK ( \rho ) \xi
=\mathrm{D}_{\xi}\mathcal{R}^* ( \rho,\xi ) , \quad\left\langle
  \xi,\bbK ( \rho ) \xi\right\rangle =2\mathcal{R}^* ( \rho,\xi )
\geq0.
\]
\end{proposition}
\begin{proof} We first observe that relation \eqref{eq:TensorGS}
  follows by employing Theorem \ref{th:MiracleIdentity}. 
  The positivity of $\calR^\ast$ is a special case of Proposition
  \ref{pr:genMira} for $\alpha=0$, i.e.\
  $\CCC_{\rho\oti\wh\sigma}\geq 0$. Finally, we obtain
  $\bbK(\rho)\xi=\rmD_\xi\calR^*(\rho,\xi)$ by noting that
  $\trace_{\HH_2} $ is adjoint to $A\mapsto A\oti
  \bm1_{\HH_2}$. This then shows $\bbK=\bbK^*\geq 0$.  
\end{proof}

Note that the operators $\mathbb{Q}$ and $\wh\sigma$ strongly depend
on $H$ and $\beta$ since they have to satisfy the commutation relation
$\left[\mathbb{Q},\RHObeta \oti \wh\sigma\right]=0$. The relation
\eqref{eq:TensorGS} shows that the above gradient structure leads
indeed to DBC Lindbladians. A direct corollary of the above gradient
structure are the Onsager operators for the building blocks
$\MM{\beta}{Q}$ introduced in Proposition \ref{pr:BuildBlock}.

\begin{corollary}[Simple Onsager operators]\label{co:DissBlockNeu}
For $H$ and $\RHObeta$ as above and $(\omega,Q)\in \mathfrak E(H)$,
we define $\KK{\beta}{Q}(\rho): \Herm{\HH} \to \Herm{\HH}$ and $\calR^*_{\beta,Q}:\Herm{\HH} \to
{[0,\infty[}$ as follows:
\[
 \KK{\beta}{Q} (\rho)\xi:=\big[Q^*,\DDD^{-\beta\omega}_\rho[Q,\xi]\big] + 
            \big[Q,\DDD^{\beta\omega}_\rho[Q^*,\xi]\big]           
\]
Then $\KK{\beta}{Q}$ is an Onsager operator as in Proposition
\ref{pr:TensorGS} and satisfies the identity
\begin{equation}
  \label{eq:bbK-QNeu}
\MM{\beta}{Q} \rho= - \KK{\beta}{Q}(\rho)\big(\log \rho+\beta H) 
\end{equation}
\end{corollary}
\begin{proof}
Choose\[
\mathbb{Q}= \block{0}{Q^\ast}{Q}{0} \quad \wh\sigma =
\block{\ee^{\beta \omega/2}}00{\ee^{-\beta \omega/2}}. 
\] Then equation \eqref{eq:bbK-QNeu} follows directly from
\eqref{eq:TensorGS} and Lemma \ref{le:CCCab}.
\end{proof}

Our main result concerning the representation of general DBC
Lindbladians is now simply stated by collecting the previous
results. We have two forms, the first is based on the compact tensor
representation and the second is based on the additive form $
\LLL=\sum_{m=1}^M \MM{\beta}{Q_m}$ in terms of the building blocks
$\MM{\beta}{Q_m}$, which will be reflected in an additive structure
for the Onsager operator $\bbK$, whereas the relative entropy as the
driving functional is independent of the $M$ different dissipative
mechanisms.

\begin{theorem}[Gradient structure for $\LLL$ with DBC] \  Consider
  $H \in \Herm{\HH}$ and $\RHObeta$ as
  above. Then, for any Lindblad operator $\LLL$ satisfying the DBC
  \eqref{eq:DBC} there exists an Onsager operator $\bbK$ such that
  $\LLL$ can be written as $\bbK$-gradient of the relative entropy
  $\calF_\beta = \calH(\cdot|\RHObeta)$. 
\\
(1) If $\LLL\rho= - \trace_{\HH_2}\big(
[\bbQ,[\bbQ,\rho\oti \RHObeta]]= - \trace_{\HH_2}\big(
[\wt\bbQ,[\wt\bbQ,\rho\oti \RHObeta^{-1}]]$ with
$\wt\bbQ=\YYY_{\RHObeta}\bbQ$, we can 
choose
\[
\bbK(\rho)\xi =\trace_{\HH_2}\Big( \big[\bbQ,\CCC_{\rho\oti\RHObeta}
[\bbQ, \xi \oti\bm1_{\HH_2}] \big]\Big) \ = \ 
  \trace_{\HH_2}\Big( \big[\wt\bbQ,\CCC_{\rho\oti\RHObeta^{-1}}
[\wt\bbQ, \xi \oti\bm1_{\HH_2}] \big]\Big).
\] 
(2) If $\LLL= \sum_{m=1}^M \MM{\beta}{Q_m}$ with
$(\omega_m,Q_m) \in \mathfrak E(H)$, we can choose 
$ \bbK(\rho)= \sum_{m=1}^M \KK{\beta}{Q_m}(\rho)$.
\end{theorem}
\begin{proof} The result (1) follows by combining the tensor representation
  in Theorem \ref{th:CompRepDBC} and Proposition \ref{pr:TensorGS}. 
The result (2) follows by combining the additive representation in
Proposition \ref{pr:BuildBlock}(c) and Corollary \ref{co:DissBlockNeu}.
\end{proof}

\section{Dissipative quantum mechanics via  GENERIC}
\label{s:GENERIC}

GENERIC is an acronym for General Equations for Non-Equilibrium
Reversible Irreversible Coupling, which was introduced by \"Ottinger
and Grmela in \cite{GrmOtt97DTCF1,OttGrm97DTCF2}. It describes a
thermodynamically consistent way of coupling Hamiltonian (=reversible)
dynamics with gradient-flow (irreversible) dynamics. It is a variant
of metriplectic systems introduced in \cite{Morr84BFIC,Morr86PJHD},
see also \cite{Morr09TBDO}. We refer to  \cite{MieTho12?GPTT}  
for an introductory survey of this framework and applications in a
large variety of applications. After our general introduction in
Section \ref{su:GENERIC-g} we will mainly dwell on
the quantum mechanical papers \cite{Otti10NTQM,Otti11GTDQ,Miel13DQMU}.    

\subsection{General setup of GENERIC}
\label{su:GENERIC-g}

A GENERIC system is defined in terms of a quintuple
$(\bfQ,\calE,\calS,\bbJ,\bbK)$, where the smooth functionals $\calE$ and
$\calS$ on the state space $\bfQ$ denote the total
energy and the total entropy, respectively. Moreover, $\bfQ$ carries
two geometric structure, namely a Poisson structure $\bbJ$ and 
a dissipative structure $\bbK$, i.e.,
for each $q\in \bfQ$ the operators $\bbJ(q)$ and $\bbK(q)$ map the
cotangent space $\rmT_q^* \bfQ$ into the tangent space
$\rmT_q\bfQ$. The evolution of the system is given by the
differential equation
\begin{equation}\label{eG.1}
\dot q= \bbJ(q) \rmD \calE(q) + \bbK(q) \rmD \calS(q),
\end{equation}
where $ \rmD \calE$ and $\rmD \calS$ are the differentials taking values in
the cotangent space.  

The basic conditions on the geometric structures $\bbJ$ and $\bbK$
are the symmetries 
\begin{subequations}
  \label{eq:JKES}
\begin{equation}\label{eG.JKsym}
\bbJ(q)=-\bbJ(q)^* \ \text{ and } \ \bbK(q)=\bbK(q)^*
\end{equation}
and the structural properties
\begin{equation}\label{eG.JKstruc}
\begin{aligned}
&\bbJ \text{ satisfies Jacobi's identity},\\
&\bbK(q) \text{ is positive semi-definite, i.e., } 
\langle \xi,\bbK(q)\xi\rangle \geq
0. 
\end{aligned}
\end{equation}
Thus, the triples $(\bfQ,\calE,\bbJ)$ and $(\bfQ,\calS,\bbK)$ form a
Hamiltonian and an Onsager or gradient system, respectively, with
evolution equations $\dot q = \bbJ(q)\rmD\calE(q)$ and $\dot q=\bbK(q)
\rmD\calS(q)$, respectively. 
Finally, the central condition states that the energy functional does not
contribute to dissipative mechanisms and that the entropy functional does not
contribute to reversible dynamics, which is the following
\emph{non-interaction condition (NIC)}:
\begin{equation}\label{eG.orth}
\forall \, q \in \bfQ:\quad \bbJ(q)\rmD\calS(q)=0 \quad \text{and} \quad 
\bbK(q)\rmD\calE(q)=0. \medskip
\end{equation}
\end{subequations}
A first observation is that \eqref{eq:JKES}
implies energy conservation and entropy increase:
\begin{align}
\label{eq:EnergBal}
\frac{\rmd}{\rmd t} \calE(q(t))&=\langle \rmD\calE(q),\dot q\rangle = 
\langle \rmD\calE(q),\bbJ\rmD\calE + \bbK\rmD\calS\rangle = 0+0=0,\\
\label{eq:EntropyProd}
\frac{\rmd}{\rmd t} \calS(q(t))&=\langle \rmD\calS(q),\dot q\rangle = 
\langle \rmD\calS(q),\bbJ\rmD\calE + \bbK\rmD\calS\rangle = 0+\langle
\rmD\calS,\bbK\rmD\calS\rangle \geq 0.
\end{align}
Note that we would need much less than the three conditions
\eqref{eq:JKES} to guarantee these two
properties. However, the next property needs \eqref{eG.orth} in its
full strength.

Namely we show that equilibria can be obtained by the \emph{maximum
 entropy principle}.  If $x_\eq $ maximizes $\calS$ under the
constraint $\calE(q)=E_0$, then we obtain a Lagrange multiplier
$\lambda_\eq \in \R$ such that $\rmD \calS(q_\eq )=\lambda_\eq
\rmD\calE(q_\eq )$. Assuming $\lambda_\eq \neq 0$ we immediately
find that $x_\eq $ is an equilibrium of \eqref{eG.1}. Indeed,
\[
\bbJ(q_\eq )\rmD\calE(q_\eq )=\tfrac1{\lambda_\eq } \bbJ(q_\eq
)\rmD\calS(q_\eq )=0 \text{ and } \bbK(q_\eq )\rmD\calS(q_\eq
)=\lambda_\eq \bbK(q_\eq )\rmD\calE(q_\eq )=0,
\]
where we used the NIC \eqref{eG.orth}. 

Vice versa, for every steady state $q_\eq $ of \eqref{eG.1} we must have 
\begin{equation}
  \label{eq:SteadyState}
  \bbJ(q_\eq )\rmD \calE(q_\eq )=0 \text{ and }  \bbK(q_\eq )\rmD
  \calS(q_\eq )=0.   
\end{equation}
Thus, in a steady state there cannot be any balancing between
reversible and irreversible forces, both have to vanish independently.
To see this we recall the entropy production relation
\eqref{eq:EntropyProd}, which implies $\langle \rmD\calS(q_\eq
),\bbK(q_\eq )\rmD\calS(q_\eq )\rangle = 0 $ for any steady state.
Since $\bbK(q_\eq )$ is positive semidefinite, this implies the second
identity in \eqref{eq:SteadyState}. The first identity then follows
from $\dot q\equiv 0$ in \eqref{eG.1}.
\medskip

Very often one is only interested in isothermal systems with fixed
temperature $\theta_*>0$, where the
free energy $\calF(q)=\calE(q)-\theta_*\calS(q)$ is a Liapunov
function. The associated structure is then that of a damped
Hamiltonian system, namely 
\begin{equation}
  \label{eq:DampedHam}
  \dot q = \Big(\bbJ(q) - \frac1{\theta_*} \bbK(q)\Big)\rmD \calF(q),
\end{equation}
where again $\bbJ$ and $\bbK$ are Poisson and Onsager structures,
respectively. However, there are no longer any non-interaction
conditions, since only  one functional $\calF$ is left. 

As in \cite{Miel11FTDM,DuPeZi13GFVF} we note that \eqref{eq:DampedHam}
can be converted to a GENERIC system by 
introducing a scalar slack variable $e$ and defining $  
\wt\calE(q,e) = \calF(q)+ e$, $\wt\calS(q,e)= \frac{e}{\theta_*}$, 
\begin{align*}
\wt\bbJ(q,e)= \block{\bbJ(q)}000, \quad \text{and} \quad \wt\bbK(q,e)=
\block{\bbK}{{-}\bbK\rmD\calF}{{-}(\bbK\rmD\calF)^\top} 
{\langle\rmD\calF,\bbK\rmD\calF\rangle}.
\end{align*} 
Clearly, the NIC \eqref{eG.orth} are satisfied. The variable $e$ can
be seen as the entropic part of the energy. Of course, in concrete
cases it is usually easy to find a physically more reasonable
splitting into entropy and energy.

\subsection{Coupling a dissipative and a quantum system}
\label{su:DissEvol}

Since GENERIC systems are closed systems with energy conservation and
entropy increase, we need to model all couplings to the quantum system
by suitable macroscopic variables. For this aim we introduce the
macroscopic variables $z$ lying in a Hilbert space $Z$. The macro-variable $z$ 
may include Hamiltonian parts (like the Maxwell equations) as well as
dissipative parts producing entropy. The important point of this work
is the thermodynamically consistent coupling of the macroscopic system
to the quantum system in such a way that energy can be exchanged via
Lindblad-like terms. 

Following \cite{Otti10NTQM,Otti11GTDQ,Miel13DQMU} we consider an
energy and an entropy in the decoupled form 
\begin{equation}
  \label{eq:E-S.decoup}
  \calE(\rho,z)=\trace(\rho H) + E(z) \quad 
\text{and} \quad \calS(\rho,z)=-\KB \trace(\rho \log \rho) + S(z).
\end{equation}
Whenever suitable we abbreviate the full state with $q=(\rho,z)$ and
choose a Poisson 
structure as follows. Consider a constant macroscopic Poisson operator
$\bbJ_\ma(z):Z^*\to Z$ and a constant coupling
operator $\Gamma: L(\HH ) \to Z $, then 
\begin{equation}
  \label{eq:bbJ.qs.ma}
  \bbJ(q)= \block{\bbJ_\qs(\rho)}{ -\bbJ_\qs(\rho)\Gamma^*}
 {-\Gamma \bbJ_\qs(\rho)}{\bbJ_\ma{+}\Gamma \bbJ_\qs(\rho)\Gamma^*} 
\text{ with } \bbJ_\qs(\rho)\mu=\ii[\rho, \mu],
\end{equation}
is a Poisson structure, which easily follows by transforming the
decoupled structure $\diag(\bbJ_\qs,\bbJ_\ma)$ via the linear
mapping  $(\rho,z)\mapsto (\rho,
z{-}\Gamma\rho)$. 

Using the relation $\rmD\calS(q)=(-\KB \log \rho, \rmD_zS(z))$ and
$[\rho,\log \rho]=0$, the first NIC 
$\bbJ(q)\rmD\calS(q)\equiv 0$ follows by asking 
\begin{equation}
  \label{eq:NIC.Gamma}
\bbJ_\ma(z)\rmD_z S(z) \equiv 0 \quad \text{and} \quad   \Gamma^*
\rmD_z S(z)  \equiv  0.  
\end{equation}

The choice for the Onsager operator $\bbK$ is more delicate, since we
do not want to generate nonlinear (non-smooth) terms arising from
$-\KB \log \rho$ in the term $\bbK(q)\rmD\calS(q)$. This will be
achieved by using the theory from above
concerning the gradient structures for the Lindblad equation for a fixed 
quantum Hamiltonian $H$ and coupling operators $Q_c$, where $c \in C$
is a finite set of couplings. Throughout we assume that
$(\omega_c,Q_c) \in \mathfrak E(H)$ are fixed, while the inverse 
coupling temperatures $\beta_c(z)$ may depend on the state of the
macroscopic system.  

More precisely  we use the ansatz
\begin{equation}
  \label{eq:calP*}
\begin{aligned}  \calP^*(\rho,z\:;\: \mu,\zeta)&= \frac12 \langle
  \zeta\, , \, \bbK_\ma(z) \zeta \rangle_Z \\
&\quad  + \sum_{c\in C} 
\frac12\SPP{ \mu {-} \langle \zeta,b_c(z)\rangle
H}{\wh\bbK_c(z)\big( \mu {-} \langle \zeta,b_c(z)\rangle
H \big)}.
\end{aligned}
\end{equation}
for the dual entropy-production potential of the coupled
system. Here $\bbK_\ma(z):Z^*\to Z$ is a symmetric and positive semi-definite
macroscopic Onsager operator. The coupling vectors $b_c(z)\in Z$ are
chosen to satisfy the conditions
\begin{equation}
  \label{eq:b.c}
  \forall\, c \in C\ \forall\,z\in Z: \quad \langle
\rmD_zE(z), b_c (z)\rangle_Z =1 \text{ and } 
\langle \rmD_z S(z), b_c(z)\rangle_Z >0
\end{equation} 
where the first ensures the NIC $\bbK(\rho,z)\rmD\calE(\rho,z)\equiv
0$ and the second the positivity of the temperature.  The Onsager
operators $\wh \bbK_c(z)$, which couple the quantum system via
the vector $b_c$ to the macroscopic system, are constructed with
the help of the operators $\KK{\beta}{ Q}(\rho):L(\HH ) \to L(\HH ) $
from Corollary \ref{co:DissBlockNeu} as follows:
\begin{equation}
  \label{eq:bbK-beta.c}
  \wh \bbK_c(z) = \kappa_c(z) \KK{\wh\beta_c(z)}{Q_c}(\rho), \quad 
   \text{where }\wh\beta_c(z):=\frac1{\KB} \langle
  \rmD_z S(z), b_c(z)\rangle_Z \text{ and }\kappa_c(z)\geq 0.
\end{equation}
Hence, the full Onsager operator takes the form 
\begin{equation}
  \label{eq:bbK.qs.ma}
  \bbK(q)= \block000{\bbK_\ma(z)}
 + \sum_{c \in C}  
 \block{\wh\bbK_c}{{-}\langle\Box,b_c\rangle_Z \wh\bbK_c H}
{{-}\SPP{H}{\wh\bbK_c\Box}b_c}{\ \SPP H{\wh\bbK_c H}b_c{\oti} b_c},
\end{equation}
where $\Box$ indicates where  the corresponding argument has
to be inserted. Using the definition of $\wh\beta_c$ in
\eqref{eq:bbK-beta.c} we easily see that the NIC
$\bbK(\rho,z)\rmD\calE(\rho,z)\equiv 0$ holds, if we assume
$\bbK_\ma(z)\rmD_z E(z)\equiv 0$. We emphasize that $\bbK$ depends
highly nonlinearly on $\rho$ through $\KK{\beta_c(z)}{Q_c}(\rho)$,
which in turn depends on $\DDD^{\pm \wh\beta_c(z)\omega_c}_\rho$. 

\begin{proposition}[GENERIC structure]\label{pr:GENERIC}
Let $X=\RRR\ti Z$ and let $\calE$ and $\calS$ be given in the decoupled form
\eqref{eq:E-S.decoup}. Moreover, consider the Poisson structure $\bbJ$
defined in \eqref{eq:bbJ.qs.ma} and the Onsager structure $\bbK$
defined in \eqref{eq:bbK.qs.ma}. Assuming additionally \eqref{eq:b.c},
$\bbJ_\ma \rmD S\equiv 0$, $\Gamma^*\rmD S \equiv 0$, and $\bbK_\ma
\rmD E \equiv 0$ the quintuple $(X,\calE,\calS,\bbJ,\bbK)$ forms a
GENERIC system. 
\end{proposition}

So far, we have not used the special form of $\wh\bbK_c$ defined in
terms of $\KK{\beta_c}{Q_c}$. The advantage of this choice is of
course dictated by the aim to obtain a linear Lindblad equation for
$\rho$. Indeed, using the relation \eqref{eq:bbK-QNeu} relating
$\KK{\beta}{Q}$ and $\MM{\beta}{Q}$ we see that the equations
obtained from the GENERIC system have the explicit form
\begin{equation}
  \label{eq:Coupl.Syst}
\begin{aligned}
  \binom{\dot\rho}{\dot z} &=
  \bbJ(\rho,z)\rmD\calE(\rho,z)+ \bbK(\rho,z) \rmD\calS(\rho, z)\\
  &= \block{\bbJ_\qs(\rho)}{ -\bbJ_\qs(\rho)\Gamma^*} {-\Gamma
    \bbJ_\qs(\rho)}{\bbJ_\ma{+}\Gamma \bbJ_\qs(\rho)\Gamma^*}
  \binom{H}{\rmD E(z)} +
  \binom{0}{\bbK_\ma(z)\rmD S(z)} \\
  &\quad +\sum_{c\in C} \block{\wh\bbK_c}{{-}\langle\Box,b_c\rangle_Z
    \wh\bbK_c H} {{-}\SPP{H}{\wh\bbK_c\Box}b_c}{\ \SPP H{\wh\bbK_c
      H}b_c{\oti} b_c}
  \binom{-\KB \log \rho}{\rmD S(z)}.
\end{aligned}  
\end{equation}
Introducing the effective Hamiltonian $\wt H(z)$ through
\[
\wt H(z) = H - \Gamma^* \rmD E(z).
\]
and using the construction of $\wh\bbK_c(\rho,z)$ via
$\KK{\wh\beta_c(z)}{Q_c}(\rho)$ we arrive at a coupled system for
$\rho$ and $z$ that is indeed linear in $\rho$, namely
\begin{equation}
  \label{eq:Coupl.rho.z}
\begin{aligned}
  \binom{\dot\rho}{\dot z} \ = \ \binom{\ii \big[ \rho\,,\, \wt
    H(z)\big]}{\bbJ_\ma\rmD E(z) -\ii \Gamma\big[ \rho\,,\, \wt
    H(z)\big] } &+ \binom{0}{\bbK_\ma(z)\rmD S(z)} \\ & + \sum_{c\in
    C}\binom{\KB \kappa_c(z) \MM{\beta_c(z)}{Q_c}\rho}{-\KB
    \kappa_c(z) \SPP{H}{ \MM{\beta_c(z)}{Q_c}\rho} b_c(z)}.
\end{aligned}
\end{equation}
Moreover, the coupling of the linear quantum system for $\rho$ with
the macroscopic system for $z$ is given in a very particular manner
reflecting the DBC, as the vectors $b_c(z)$ occur twice, namely
(i) in the definition of $\beta_c(z) = \langle \rmD
S(z),b_c(z)\rangle/\KB$ in \eqref{eq:bbK-beta.c} and (ii) in the equation
for $z$, i.e.\ the second component of \eqref{eq:Coupl.rho.z}.
The fact that this equation is obtained from the GENERIC system
$(\RRR\ti Z,\calE,\calS,\bbJ,\bbK)$ implies energy
conservation and entropy production along solutions
$q(t)=(\rho(t),z(t))$, namely $ \frac\rmd{\rmd t} \calE(q(t)) \equiv 0$ and 
\begin{align*}
\frac\rmd{\rmd
  t} \calS(q(t)) &= 2\calP^*(q(t),\rmD\calS(t)) \\
&= 2\langle \rmD
S,\bbK_\ma \rmD S\rangle_Z + 2 \KB^2 
\sum_{c\in C} \SPP{\log \rho {+}\wh\beta_c H}{\KK{\wh\beta_c}{Q_c}(\log \rho
  {+}\wh\beta_c H)}  \geq 0. 
\end{align*} 
\begin{comment}
For these relations to be true it is essential to choose $\wh\bbK_c$
and hence $\KK{\beta}{Q} $, in a very particular way. Moreover, it is
crucial to normalize $b_c(z)$ and $\wh\beta_c(z)$ in such a way that
\[
\langle \rmD E(z),b_c(z)\rangle_Z = 1 \quad \text{and} \quad
\langle \rmD S(z),b_c(z)\rangle_Z =\KB \wh\beta_c(z),  
\]
see \eqref{eq:b.c} and \eqref{eq:bbK-beta.c}. 
\end{comment}

\section{Examples and applications}
\label{se:ExaAppl}

We discuss the above construction and give a few examples and
applications to highlight the concept of the operators
$\DDD^\alpha_\rho$ and the relevance of the corresponding Lindblad
operators. From a modeling point of view, we emphasize that using
the gradient structure with respect to the relative entropy, it is easy
to construct thermodynamically consistent coupled system including
macroscopic variables $z\in Z$ and a quantum state $\rho$, see e.g.\
the quantum-dot model discussed in Section \ref{su:QD}. In particular,
we can interprete the dissipative mechanisms in the macroscopic
system and in the quantum system as given building blocks that have to
be combined in a suitable way to obtain a thermodynamically correct
system. This can either be a damped Hamiltonian system at constant
temperature or a GENERIC (General Equation for Non-Equilibrium
Reversible Irreversible Coupling) system where the total energy is
preserved while the physical entropy increases. 

\subsection{The isothermal, damped quantum system}
\label{su:DamQS}

For a general DBC Lindbladian $\LLL$ with respect to $\RHObeta=\frac1Z
\ee^{-\beta H}$ Proposition \ref{pr:BuildBlock} shows that it can be
written as a sum of operators $\MM{\beta}{Q}$ as defined in
Proposition \ref{pr:BuildBlock}(b). Indeed, we have
\begin{equation}
  \label{eq:DampHS}
  \dot \rho= \ii [\rho, H] + \LLL \rho = \ii[\rho,H] + \sum_{c\in C}
  \MM{\beta}{Q_c} \rho, 
\end{equation}
where the coupling operator $Q_c$ satisfy
$(\omega_c,Q_c) \in \mathfrak E(H)$. 

We are now able to state that all dissipative quantum generators
satisfying the DBC can be written as a damped Hamiltonian system
$(\RRR, \calF, \bbJ, \bbK)$, namely 
\[
\dot \rho = \ii [\rho, H] + \LLL \rho = \big( \bbJ_\qs(\rho) -
\bbK(\rho)\big) \rmD \calF(\rho),
\]
where we can use the following choices 
\begin{align*} 
&\calF(\rho) = \trace\big(\beta H + \rho \log \rho\big) ,\quad  
\bbJ_\qs(\rho)= \ii\big[ \rho,\Box\big], \quad \text{and} \quad
\bbK(\rho) = \sum_{c \in C} \KK{\beta}{Q_c}(\rho) . 
\end{align*}
For this we simply use that $(\omega,Q)\in \mathfrak E(H)$ implies 
 $\KK{\beta}{Q} (\log \rho{+}\beta H) = \MM{\beta}{Q}
\rho$, see \eqref{eq:bbK-QNeu}.

\subsection{Coupling to simple heat baths}
\label{su:SimpleBath}

We consider a quantum system coupled to a finite number of
finite-energy heat baths indexed by $m=1,..,M$. We set $Z= \R^M$ with
elements $z=\bftheta= (\theta_m)_{m}$, where $\theta_m>0$ denotes the
absolute temperature of the $m$th heat bath. Each heat bath is
coupled to the quantum state $\rho$ such that energy may flow from one
heat bath into the others via the quatum system. 

We assume each heat bath
to have a constant specific heat $c_m$. 
Hence, we let
\[
E(\bftheta )= \sum_{m=1}^M c_m \theta_m \quad \text{and} \quad
S(\bftheta) = \sum_{m=1}^M c_m \log \theta_m. 
\]
For the coupling vectors $b_m$ we choose 
\[
b_m(\bftheta) = \frac1{c_m} \bfe^{(m)}\in \R^M, \quad \text{where
}\bfe^{(m)} =(0,..,0,1,0,..,0)^\top
\]
is the $m$th unit vector. Clearly we find 
\[
b_m(\bftheta)\cdot \rmD E(\bftheta)\equiv 1 \ \ \text{ and } \ \  
\wh\beta_m(\bftheta)= \frac1{\KB} b_m(\bftheta) \cdot \rmD
S(\bftheta)= \frac{1}{\KB \theta_m}, 
\]
which is the usual inverse temperature of the $m$th heat bath.  

Assuming $\bbJ_\ma\equiv 0$ and $\Gamma=0$, we can choose a symmetric
and positive semi-definite matrix
$K_\ma\in \R^{M\ti M}$ such that $K_\ma \ol\bfc=0$, where $\ol\bfc$ is
the constant vector $\rmD E(\bftheta)=(1/c_m)_{m=1,..,M}$. The
construction in Section \ref{su:DissEvol} provides a
GENERIC system for $q=(\rho,\bftheta)$ in the following form:
\begin{equation}
  \label{eq:QS.HeatB}
 \begin{aligned}  \dot\rho&= \ii[\rho, H] + \sum_{m=1}^M \MM{1/(\KB \theta_m)}{Q_m}
  \rho ,\\ \dot \bftheta &= K_\ma \rmD S(\bftheta) + \sum_{m=1}^M
  \frac1{c_m} \SPP{H}{\MM{1/(\KB \theta_m)}{Q_m} \rho}\,\bfe^{(m)}. 
\end{aligned}
\end{equation}
It is now easy to see that we may take the heat capacities very large,
such that the temperatures $\theta_m$ do not change any more, or at
least not on the time scale where the typical changes of $\rho$ occur. 
Thus, it is possible to investigate in a natural way non-equilibrium
steady states, where energy is exchanged between heat baths with
different temperatures. 

\subsection{An isothermally coupled system}
\label{su:IsoQS.z}

Here we return to an isothermally system where the quantum
state $\rho$ 
is coupled to a macroscopic variable $z$ in such a way that we obtain a
damped Hamiltonian system. In contrast to the simple model in Section
\ref{su:DamQS} we now allow the effective inverse temperatures
$\wh\beta_c$ to depend on the state $z$. Here we give a general
form of such systems and in Section \ref{su:QD} we provide a model for
a quantum dot interaction with charge-carrier densities taking the
roles of $z$. 

For the definition
of the free energy $\calF$ we choose the fixed equilibrium
temperature $\theta_*$ and obtain
\[
\calF(\rho,z) = \frac1{\KB \theta_*} \calE(\rho,z) - \frac1{\KB}
\calS(\rho,z) = \trace\big(\rho \log \rho + \frac1{\KB\theta_*}
 \rho H\big) + F(z). 
\] 
We may take $\bbJ$ as in \eqref{eq:bbJ.qs.ma} but need to choose a
special $\bbK$ to obtain equations linear in $\rho$, namely
\[
\bbK(\rho,z)= \block000{\bbK_\ma(z)} + \sum_{c\in C}
\block{\wt\bbK_c }{\langle\Box,a_c\rangle_Z\,\wt\bbK_c H}
{\SPP{H}{\wt\bbK_c\Box}a_c}{ \SPP{H}{\wt\bbK_c H}
  a_c{\oti}a_c}
\]
with coupling vectors $a_c(z)\in Z$ and Onsager operators
$\wt\bbK_c(\rho,z)=\kappa_c(z) \KK{\wt\beta_c(z)}{Q_c}(\rho)$, where
\[
\kappa_c(z)\geq 0 \quad \text{and} \quad \wt\beta_c(z)= 
\frac1{\KB\theta_*}  + \langle \rmD
F(z),a_c(z)\rangle_Z>0. 
\]
Hence, the equations $\dot q=\big( \bbJ(q){-}\bbK(q)\big) \rmD\calF(q)$
take the form 
\begin{align*}
\dot\rho&= \ii \big[\rho, H{-}\Gamma^* \rmD F(z)\big] \hspace*{6.2em}+
\sum_{c\in C} 
\kappa_c(z) \MM{\wt\beta_c(z)}{Q_c} \rho ,\\
\dot z&= \bbJ_\ma \rmD F(z) - \ii \Gamma\big[\rho, H{-}\Gamma^* \rmD
F(z)\big] +\sum_{c\in C} \kappa_c(z)\SPP{H}{\MM{\wt\beta_c(z)}{Q_c}
  \rho} a_c(z).    
\end{align*}
Again $F(Z)$, $a_c(z)$, and $\wt\beta_c(z)$ are 
intrinsically linked to each other in order to generate this
evolutionary equation from 
the damped Hamiltonian system $(\RRR{\ti}Z,\calF,\bbJ,\bbK)$.

\subsection{A quantum dot interacting with charge
  carriers}
\label{su:QD}

In \cite{RGGJ10EPPS} a four-level model with two levels in the
conduction and two levels in the valence band is described in
detail. However, the charge carriers in the wetting layer are assumed
to be given, i.e.\ kept constant during the experiment. In
\cite{LuMaSc09RDEH}, the turn-on behavior of single quantum-dot lasers
was analyzed based on a rate equation for the photon density
$n_\text{ph}$, the densities $w_\rme$ and $w_\rmh$
for electrons and holes in the wetting layer, and the charge carriers
$n_\rme$ and $n_\rmh$ in the quantum dot. We will show in
\cite{MiMiRo17?TCCQ} that the approach developed in Section
\ref{su:IsoQS.z} allows us to combine theses two approaches to derive
a thermodynamically consistent model coupling the charge-carrier
relaxation with the state of the quantum dot.

Here we look at a very simplistic model, where the quantum system has
only two states with energies $\eps_2>\eps_1$. Moreover, in the
wetting layer we have free charge
carriers with density $c_\rmf$ that can be captured by the quantum
dot leading to a density $c_\rmb$ of bound charge carriers. This charge
capture excites 
the quantum state from $|1\rangle$ to $|2\rangle$: 
\[
X_\text{free} + |1\rangle \quad \rightleftharpoons \quad
X_\text{bound} + |2\rangle.
\]   
This is one example of several capture-escape or scattering processes
that need to be modeled for a complete description of single
quantum-dot lasers, see \cite{RGGJ10EPPS}.
 
With $H=\diag(\eps_1,\eps_2)\in \Herm{\C^2}$, a fixed $\beta_*>0$,
and the density vector $z=\bfc=
(c_\rmf,c_\rmb)\in {]0,\infty[}^2$, we have the relative entropy   
\[
\calF(\rho, z)= \trace(\rho \log \rho + \beta_* \rho H) + c_\rmf \big(
\log(c_\rmf/w_\rmf)-1\big) + c_\rmb \big(
\log(c_\rmb/w_\rmb)-1\big), 
\]
where $w_\rmf,w_\rmb>0$ are suitable equilibrium densities. With 
\[
\bfa=\frac1{\eps_2{-}\eps_1} \binom{-1}{1} \ \text{ and } \ 
Q= |1\rangle \;\! \langle 2| = \bma{cc} 0&1\\ 0&0\ema 
\]
we define the Onsager operator as in Section \ref{su:IsoQS.z}: 
\begin{align*}
&\bbK(\rho,\bfc)= \kappa(\bfc) \bma{cc} 
\calK_Q^{\wt\beta(\bfc)}(\rho) & 
\langle \bfa,\Box\rangle_Z\: \calK_Q^{\wt\beta(\bfc)}(\rho)H  \\
\SPP{H}{\calK_Q^{\wt\beta(\bfc)}(\rho)\Box}\bfa &
\SPP{H }{\calK_Q^{\wt\beta(\bfc)}(\rho) H} \bfa\oti\bfa  \ema,
 \\ 
&\text{where } 
\kappa(\bfc)=\wh\kappa\big(\frac{c_\rmf c_\rmb}{w_\rmf w_\rmb})^{1/2}
\text{ and } 
\wt\beta(\bfc)= \beta_* - 
 \frac1{\eps_2{-}\eps_1} \big( \log(c_\rmf/w_\rmf) - \log(c_\rmb/w_\rmb)\big).
\end{align*} 
Here $\wt\beta (\bfc)$ and $\bfa$ are chosen in a very particular way
such that
\begin{align*}
&-\bbK(\rho,\bfc)\rmD \calF(\rho,\bfc)= -\bbK(\rho,\bfc)\binom{\log
  \rho+\beta_*H}{ (\log c_j/a_j)_{j\in \{\rmf,\rmb\}}}\\
&=-\kappa(\bfc)
\binom{\calK_Q^{\wt\beta(\bfc)}(\rho)\big( \log \rho +
  \wt\beta(\bfc)H\big)}{ \SPP{H}{\calK_Q^{\wt\beta(\bfc)}(\rho)\big( \log \rho +
  \wt\beta(\bfc)H\big)} \bfa}=
\kappa(\bfc)\binom{\calM^{\wt\beta(\bfc)}_Q \rho}{
  \SPP{H}{\calM^{\wt\beta(\bfc)}_Q \rho} \bfa } ,
\end{align*}
where we used the miracle identity as stated in  Corollary
\ref{co:DissBlockNeu}.

Using the explicit form of $\calM^b_Q$, $H$, $Q$, and
$\wt\beta(\bfc)$ we
obtain (with $\omega=\eps_2{-}\eps_1$) 
\begin{align*}
&\calM_Q^{\wt\beta(\bfc)}\rho =  \ee^{\beta_*\omega/2} 
\big(\frac{c_\rmb}{a_\rmb}\frac{a_\rmf}{c_\rmf} \big)^{1/2} \calN_Q \rho
+   
\ee^{-\beta_*\omega/2} \big(\frac{c_\rmf}{w_\rmf}\frac{w_\rmb}{c_\rmb} \big)^{1/2}
\calN_{Q^*} \rho\Big) ,\\
& \text{where } \calN_QA:=[Q,AQ^*]+[QA,Q^*]. 
\end{align*}
By our specific choice of $\kappa(\bfc)$ all the square roots cancel,
and we see that the final coupled
system takes the form (where $\wt\kappa= \wh\kappa
\ee^{\beta_*(\eps_1+\eps_2)/2}$) 
\begin{subequations}
  \label{eq:CaptEsc}
\begin{align}   \label{eq:CaptEsc.a}
\dot \rho\ \ &=\ \ii [\rho, H] + \wt\kappa\Big(
\frac{c_\rmb}{w_\rmb }\ee^{-\beta_*\eps_1}  \calN_Q \rho +
\frac{c_\rmf}{  w_\rmf }  \ee^{-\beta_*\eps_2} \calN_{Q^*}
\rho\Big), \\   \label{eq:CaptEsc.b}
\binom{\dot c_\rmf}{\dot c_\rmb}&=2\wt\kappa\,
\Big(\frac{c_\rmf}{  w_\rmf }  \ee^{-\beta_*\eps_2} \rho_{11}  
 -\frac{c_\rmb}{w_\rmb }\ee^{-\beta_*\eps_1} \rho_{22} \Big) \,\binom{-1}1 ,  
\end{align}
\end{subequations}
where we used $\SPP{H}{\calN_Q\rho}=-2\omega \rho_{22}$ and 
$\SPP{H}{\calN_{Q^*}\rho}=2\omega\rho_{11}$. Thus, we obtain a coupled
system with quadratic nonlinearities. 

Since $\calN_Q$ and $\calN_{Q^*}$ are Lindblad operators, we see that
\eqref{eq:CaptEsc.a} is a Lindblad equation that depends on the
macroscopic variables $\bfc=(c_\rmf,c_\rmb)$. The charge carrier
densities $c_\rmf$ and $c_\rmb$ are driven by the quantum state $\rho$
in such a way that $c_\rmf,c_\rmb>0$ is always preserved. The
important message of our construction is that this system is a damped
Hamiltonian system, and thus the relative entropy $\calF$ is a
Liapunov function.

\subsection{Maxwell-Bloch model}
\label{su:MaxwellB}

In this section we discuss a nonlinear PDE model, where the Maxwell
equations on $\R^3$ as the macroscopic Hamiltonian system are coupled
to a spatially localized domain $\Omega\subset \R^3$, where at each
macroscopic point $x\in \Omega$ there is a quantum system that is
coupled to the electromagnetic fields $\bfE$ and $\bfH$, but not to
the neighboring quantum systems. Such models are called Maxwell-Bloch
systems (cf.\ e.g.\ \cite{JoMeRa00TNGO,Duma05GEMB}) and are commonly
used to model the interaction of light and matter. We refer to
\cite{Miel15TCQM} where even the coupling to the drift-diffusion
system for electrons and holes is described in terms of the GENERIC
framework.

The macroscopic system is described by $z=(\bfE,\bfH)\in
Z:=\rmL^2 (\R^3;\R^3)^2$ denoting the electric and the magnetic
fields. The optically active material is described by a bounded Lipschitz
domain $\Omega \subset \R^3$, where the quantum state is described via
$\rho:\Omega \to \RRR_N\subset \Herm{\HH}$. The quantum
state determines the macroscopic polarization $\bfP:\Omega \to \R^3$ via a
polarization operator $\Gamma:L(\HH) \to \R^3$ in the form
$\bfP(x) = \Gamma \rho(x)$. The electric displacement field then is
given by $\bfD=\eps_0 \bfE+\bfP$,  and Maxwell equations take the form
\begin{equation}
  \label{eq:Maxwell}
  \eps_0\dot\bfE +\dot\bfP = \curl \bfH , \quad \mu_0
  \dot\bfH=-  \curl  \bfE,  \quad  \div \big(\epsilon_0  \bfE
  {+} \bfP\big) = 0   , \quad \div\bfH=0.  
\end{equation}

The main difficulty is to model the coupling of the quantum systems
$\rho(t,x) \in \RRR_N$ in a consistent way.  We derive the system for
$q:=(\rho,\bfE,\bfH)$ in the form of a damped Hamiltonian system with
the free energy
\[
\calF(q)= \int_{\R^d} \frac12|\bfE|^2 + \frac12|\bfH|^2
\dd x + \int_\Omega \trace\big(\rho\log \rho+ \beta H_\rmB \rho\big)\dd x,
\]
where $H_\rmB$ is called the Bloch Hamiltonian and where we have set
the constants $\eps_0$ and $\mu_0$ to $1$ for notational simplicity.
Moreover, we define a Poisson structure $\bbJ$ and an Onsager operator
$\bbK$ in the form
\[
\bbJ(q) = \bma{@{}ccc@{}} \ii[\rho,\Box] &-\ii\big[\rho,
\Gamma^*\Box\big]  &0 \\ \!\!- \ii\Gamma\big[\rho,\Box]
&\ii\Gamma\big[\rho,\Gamma^*\Box] &\curl \\ 0 &-\curl &0 \ema \ 
\text{and} \  \bbK(q)=\bma{@{}ccc@{}}
\bbK_\qs(q)&-\bbK_\qs(q)\Gamma^*&0\\
\!\!-\Gamma\bbK_\qs(q)&\Gamma\bbK_\qs(q)\Gamma^*&0\\ 0 & 0&0
\ema. 
\]
The importance of this choice for $\bbJ$ and $\bbK$ is that the first
row of each of the block operators is replicated in the second row, but
premultiplied with $-\Gamma$. This is needed to obtain the correct
first equation in the Maxwell system \eqref{eq:Maxwell}, namely
\[
\dot \bfE = \curl \bfH - \dot\bfP = \curl \bfH - \Gamma\dot\rho.
\]
By the symmetries $\bbJ=-\bbJ^*$ and $\bbK=\bbK^*$ this also implies
that the first columns are replicated in the second columns, but now
post-multiplied by $-\Gamma^*$. 

We emphasize that the optically active material is restricted to the
bounded domain $\Omega$, while the fields $\bfE$ and $\bfH$ are
defined on all of $\R^3$. This is hidden in our operator $\Gamma$,
which should be understood as an operator that maps $\rho$ defined on
$\Omega$ to vector fields $\bfP=\Gamma \rho$ that are defined not only
in $\Omega$, but that are extended by $0$ outside of
$\Omega$. Similarly, the adjoint operator $\Gamma^*$ acts on $\bfE$ by
restricting it first to $\Omega$ and then applying the adjoint map of
$\Gamma$. 

Thus, using $[\rho,\log \rho] =0$ the induced equation for $\rho$
reads 
\[
\dot \rho \ =  \ \ii \big[\,\rho,\,\beta H_\rmB - \Gamma^* \bfE\,\big]\ - \ 
\bbK_\qs(\rho,\bfE,\bfH)\big( \log\rho + \beta H_\rmB - \Gamma^* \bfE\big). 
\]
Having derived this structure it is now our task to choose $\bbK_\qs$
in such a way that the equation for $\rho$ is a Lindblad equation that
is linear in $\rho$ with coefficients depending on $\bfE$ and $\bfH$.
For this, we assume that there exists a family $(Q_c)_{c\in C}$ of
couplings such that
\begin{equation}
  \label{eq:Maxw.Qc}
  (\omega_c,Q_c) \in \mathfrak E(H_\rmB) \quad \text{ and } \quad 
  [Q_c, \Gamma^*\bfE] = \omega_c\, (\bfg_c{\cdot} \bfE)\:Q_c,
\end{equation}
for some vectors $\bfg_c\in \R^3$. Note that this is a non-trivial
condition, since it implies that for all $\bfE\in \R^3$ the matrix
$\Gamma^*\bfE\in \Herm{\HH}$ commutes with all matrices
$Q$ that commute with $H$. 

Based on the assumption \eqref{eq:Maxw.Qc} we are now able to choose
$\bbK_\qs$ in the form  
\[
\bbK_\qs(\rho,\bfE,\bfH)= \sum_{c\in C} \kappa_c(\bfE,\bfH)\,
\KK{\wt\beta_c(\bfE)}{ Q_c}(\rho) , \quad \text{where }
\wt\beta_c(\bfE):= \beta - \bfg_c{\cdot}\bfE. 
\]
Since by construction we have $[Q_c,\beta H_\rmB{-}\Gamma^*\bfE]
=\omega_c\wt\beta_c(\bfE) Q_c$  we can apply Proposition
\ref{co:DissBlockNeu} (cf.\ \eqref{eq:bbK-QNeu})  and obtain the
equation 
\[
\dot \rho = \ii \big[\,\rho,\,\beta H_\rmB - \Gamma^* \bfE\,\big]+
\sum_{c\in C} \kappa_c(\bfE,\bfH)\, \MM{\wt\beta_c(\bfE)}{Q_c} \rho,  
\]
which is linear in $\rho$ for fixed $\bfE$ and $\bfH$.  Note that for
the above constructions it is not necessary to have $\wt\beta_c \geq
0$, since $\KK{\beta}{Q}$ and $\MM{\beta}{Q}$ are well defined for all
$\beta\in \R$.  Altogether we conclude that under the rather
restrictive assumption \eqref{eq:Maxw.Qc} the damped Hamiltonian
system $(\rmL^2(\Omega;\RRR_N){\ti} Z, \calF,\bbJ,\bbK)$ generates the
following coupled Maxwell-Bloch system
\begin{align*}
\dot\rho&= \ii\big[\rho, H_\rmB {-} \Gamma^*\bfE\big] + \sum_{c\in C}
\kappa_c(q) \MM{\beta}{Q_c} \rho, \\
\dot\bfE&= \curl \bfH - \Gamma \Big(\ii\big[\rho,
H_\rmB{-}\Gamma^*\bfE \big] + \sum_{c\in C}
\kappa_c(q) \MM{\beta}{Q_c} \rho \Big), \\
\dot\bfH &= - \curl \bfE.    
\end{align*} 
It is easy to see that condition \eqref{eq:Maxw.Qc} holds, if
$\Gamma$ has the form $\Gamma \bfE=\sum_{n=1}^N
(\bfb_n{\cdot} \bfE)\: h_n\oti \ol h_n$ for arbitrary polarization
vectors $\bfb_n\in \R^3$, where $H=\sum_1^N \eps_n
h_n\oti \ol h_n$. It remains open to model the case of interaction
between different energy levels, for which $\Gamma\bfE$ needs to
contain terms like $h_m \oti \ol h_n$ with $m\neq n$.

\footnotesize
%\bibliographystyle{my_alpha}
%\bibliography{alex_pub,bib_alex}
% files available at my webpage .../people/mielke/tex/  
\newcommand{\etalchar}[1]{$^{#1}$}
\def\cprime{$'$}

\end{document}